\pdfoutput=1
\PassOptionsToPackage{protrusion, expansion}{microtype} 

\documentclass[acmsmall]{acmart}

\usepackage{adjustbox}
\usepackage{lineno}
\usepackage{algorithm}
\usepackage{algorithmic}
\usepackage{multirow}
\usepackage[tight,footnotesize]{subfigure}
\usepackage{epsfig}

\usepackage{amssymb}
\usepackage{footmisc}
\usepackage{footnpag}
\usepackage{amsmath}
\usepackage{paralist}
\usepackage{color}
\usepackage{perpage}
\usepackage{verbatim}
\usepackage{enumitem}
\usepackage{tikz}
\usepackage{listings}
\usepackage{graphicx}
\usepackage{soul}
\usepackage{fancyhdr}
\usepackage[inkscapeformat=png]{svg}
\usepackage{stfloats}
\usepackage{balance}

\newcommand*\circled[1]{\tikz[baseline=(char.base)]{
            \node[draw,circle,inner sep=0.5pt] (char) {\textcolor{black}{#1}};}}

\setcopyright{licensedusgovmixed}
\acmJournal{PACMMOD}
\acmYear{2024} \acmVolume{2} \acmNumber{1 (SIGMOD)} \acmArticle{4} \acmMonth{2}\acmDOI{10.1145/3639259}

\begin{document}

\title[High-performance Effective Scientific Error-bounded Lossy Compression with Auto-tuned Interpolation]{High-performance Effective Scientific Error-bounded Lossy Compression with Auto-tuned Multi-component Interpolation}

\author{Jinyang Liu}
\affiliation{%
  \institution{University of California, Riverside}
  \city{Riverside}
  \state{CA}
  \country{USA}}
\email{jliu447@ucr.edu}

\author{Sheng Di}
\authornote{Corresponding author}
\affiliation{%
  \institution{Argonne National Laboratory}
  \city{Lemont}
  \state{IL}
  \country{USA}}
\email{sdi1@anl.gov}

\author{Kai Zhao}
\affiliation{%
  \institution{Florida State University}
  \city{Tallahassee}
  \state{FL}
  \country{USA}}
\email{kzhao@cs.fsu.edu}

\author{Xin Liang}
\affiliation{%
  \institution{University of Kentucky}
  \city{Lexington}
  \state{KY}
  \country{USA}}
\email{xliang@cs.uky.edu}

\author{Sian Jin}
\affiliation{%
  \institution{Indiana University Bloomington}
  \city{Bloomington}
  \state{IN}
  \country{USA}}
\email{sianjin@iu.edu}

\author{Zizhe Jian}
\affiliation{%
  \institution{University of California, Riverside}
  \city{Riverside}
  \state{CA}
  \country{USA}}
\email{zjian106@ucr.edu}

\author{Jiajun Huang}
\affiliation{%
  \institution{University of California, Riverside}
  \city{Riverside}
  \state{CA}
  \country{USA}}
\email{jhuan380@ucr.edu}

\author{Shixun Wu}
\affiliation{%
  \institution{University of California, Riverside}
  \city{Riverside}
  \state{CA}
  \country{USA}}
\email{swu264@ucr.edu}

\author{Zizhong Chen}
\affiliation{%
  \institution{University of California, Riverside}
  \city{Riverside}
  \state{CA}
  \country{USA}}
\email{chen@cs.ucr.edu}

\author{Franck Cappello}
\affiliation{%
  \institution{Argonne National Laboratory}
  \city{Lemont}
  \state{IL}
  \country{USA}}
\email{cappello@mcs.anl.gov}

\renewcommand{\shortauthors}{Jinyang Liu, et al.}

\begin{abstract}

Error-bounded lossy compression has been identified as a promising solution for significantly reducing scientific data volumes upon users' requirements on data distortion. For the existing scientific error-bounded lossy compressors, some of them (such as SPERR and FAZ) can reach fairly high compression ratios and some others (such as SZx, SZ, and ZFP) feature high compression speeds, but they rarely exhibit both high ratio and high speed meanwhile. In this paper, we propose HPEZ (a.k.a. QoZ 2.0) with newly designed interpolations and quality-metric-driven auto-tuning, which features significantly improved compression quality upon the existing high-performance compressors, meanwhile being exceedingly faster than high-ratio compressors. The key contributions lie as follows: (1) We develop a series of advanced techniques such as interpolation re-ordering, multi-dimensional interpolation, and natural cubic splines to significantly improve compression qualities with interpolation-based data prediction. (2) The auto-tuning module in HPEZ has been carefully designed with novel strategies, including but not limited to block-wise interpolation tuning, dynamic dimension freezing, and Lorenzo tuning. (3) We thoroughly evaluate HPEZ compared with many other compressors on six real-world scientific datasets. Experiments show that HPEZ outperforms other high-performance error-bounded lossy compressors in compression ratio by up to 140\% under the same error bound, and by up to 360\% under the same PSNR. In parallel data transfer experiments on the distributed database, HPEZ achieves a significant performance gain with up to 40\% time cost reduction over the second-best compressor.

\end{abstract}

\ccsdesc[500]{Information systems~Data compression}
\ccsdesc[500]{Theory of computation~Data compression}
\ccsdesc[500]{Mathematics of computing~Interpolation}

\keywords{error-bounded lossy compression, interpolation, scientific database}
\settopmatter{printfolios=true}
\maketitle
\makeatletter \gdef\@ACM@checkaffil{} \makeatother
\setlength{\textfloatsep}{6pt}
\section{Introduction}
\label{sec:introduction}

The gigantic scale and exceptionally intense computation power of modern supercomputers have empowered the exascale scientific simulation applications to generate tremendous amounts of data in short periods, bringing up significant burdens for distributed scientific databases and cloud data centers. For instance, A one-trillion particle Hardware/Hybrid Accelerated Cosmology Code (HACC) \cite{habib2016hacc} can harness approximately 22PB output data in a single simulation, and Community Earth System Model (CESM) \cite{cesm} simulation may generate 2.5PB data for a simulation task \cite{straatsma2017exascale}. To this end, error-bounded lossy compression techniques have been developed for those scientific data, and they have been recognized as the most proper strategy to manage the extremely large amount of data. The advantage of error-bounded lossy compression is primarily two-fold. On the one hand, it can reduce the original data to an incredibly shrunken size which is much smaller than the compressed data size generated by a lossless compressor. On the other hand, the error-bounded lossy compression can constrain the point-wise data distortion strictly upon the users' requirements. Existing state-of-the-art error-bounded lossy compressors in diverse archetypes, such as prediction-based model -- SZ3 \cite{szinterp,sz3} and QoZ \cite{qoz}, transform-based model -- ZFP \cite{zfp} and SPERR \cite{SPERR}, and dimension-reduction-based model -- TTHRESH \cite{ballester2019tthresh}, have been widely adopted in many use cases in practice. 

Considering the abundant scope of related optimization strategies, we summarize the existing error-bounded lossy compressors as well as their pros and cons as follows. The orthogonal transform-based compressors like ZFP, exhibit high execution speeds but their compression ratios are limited to a certain extent because they focus on only local correlations (confined within $4^d$-blocks). The wavelet-based compressors such as SPERR and the Singular Value Decomposition (SVD) based compression such as TTHRESH, although can obtain quite high compression ratios, suffer from very low compression speeds attributed to their high-cost integrated data operation modules. Some prediction-based compressors (e.g. SZ3 and QoZ) deliver relatively high compression ratios with moderate running speeds, nevertheless, they may suffer from relatively low compression ratios in some cases. Recently, FAZ \cite{liu2023faz} attempted to create a hybrid framework taking advantage of heterogeneous compression techniques, however, its design fully orients the optimization of rate-distortion, so that its compression/decompression is much slower than the classic compressors such as SZ and ZFP. 

For modern scientific databases and cloud data centers which often involve multiple sites over a wide area network (WAN), the extremely large amount of raw data costs an unacceptable time to transfer between machines. Therefore, data compressors are critical for efficient data transfer because transferring compressed data will significantly reduce the time cost, as confirmed by prior research \cite{ftsz,globus-compression}. In this case, compression ratios and speeds are both critical for achieving high data transfer throughput. However, designing a versatile error-bounded lossy compressor that delivers high compression ratios with sufficient performance (i.e. speed) is quite challenging. On one hand, to reach a high compression performance, general techniques have to perform relatively simple data transform \cite{zfp} or prediction within short-range areas \cite{Xin-bigdata18,szinterp,MGARD}, which cannot take advantage of long-range data correlations, thus leading to very limited compression ratios inevitably. On the other hand, to reach a high compression ratio, general techniques are applying sophisticated techniques such as wavelet transform on the full data input \cite{liu2023faz, SPERR} or higher-order SVD \cite{ballester2019tthresh}, which suffer from very expensive operations inevitably, conflicting with our high-performance objective. As such, we must design more compact and effective data operation methods with relatively low computational costs, featuring high speed, meanwhile yielding comparable compression ratios compared to the existing high-ratio compression techniques.

In order to design an error-bounded compressor that features both high compression ratios and satisfactory speeds, we propose an optimized quality-metric-driven error-bounded lossy compressor HPEZ (also known as QoZ 2.0 as it serves as the second major version of QoZ) by developing a brand-new auto-tuning strategy and an anchor-based level-wise hybrid interpolation predictor. Integrating extensively optimized interpolation predictors and auto-tuning modules, HPEZ attains far better compression ratios and lower distortions than other high-performance error-bounded lossy compressors with limited compression speed degradation. HPEZ substantially outperforms high-ratio compressors in terms of speed. It achieves optimized throughput performance in a variety of use cases such as parallel data transfer for large (distributed) databases. We attribute our contributions as follows:
\begin{itemize}
    \item Founded on theoretical analysis and algorithmic optimizations, we substantially upgrade the most critical step in the quality-oriented compression -- interpolation prediction, leading to an immensely improved data prediction accuracy.
    \item We develop a series of optimization strategies including block-wise interpolation tuning, dynamic dimension freezing, and Lorenzo tuning, which can substantially improve the adaptability of the auto-tuning for the compression across a broad spectrum of inputs.
    \item We perform solid experiments using 6 real-world scientific datasets. HPEZ significantly outperforms state-of-the-art error-bounded lossy compressors in terms of rate-distortion, while still having a satisfactory speed. It preserves a leading speed compared to other high-ratio compressors. Consequently, it achieves the best throughput in distributed data transfer over WAN based on our experiments. HPEZ exhibits the least time cost in data transfer for most scientific datasets with up to 40\% time reduction.
\end{itemize}

The remainder of this paper is organized as follows: Section \ref{sec:related} introduces related works. Section \ref{sec:backpro} provides the research background and the research problem formulation. Section \ref{sec:framework} demonstrates the overall framework of HPEZ. The new designs of interpolation predictors in HPEZ are illustrated in detail in Section \ref{sec:interp}, and our designed auto-tuning blocks are proposed in Section \ref{sec:autotuning}. In section \ref{sec:evaluation}, the evaluation results are presented and analyzed. Finally, Section \ref{sec:conclusion} concludes our work and discusses future work. 

\section{Related Work}
\label{sec:related}

In general, scientific data compression techniques can be divided into two categories - lossless compression and lossy compression. Examples of existing lossless compressors for databases are Gorilla \cite{gorilla} and AMMO \cite{icde_time_series_compression} for time-series data, and traditional lossy data compression methods include ModelarDB \cite{jensen2018modelardb,kejser2019scalable} for time-series data and \cite{isabela,zhang2018trajectory,fang20212,li2018deep} for Geology spatial-temporal data. Besides that, error-bounded lossy compression has been preferred and crafted to serve various scientific data reduction applications \cite{use-case} and scientific databases. To meet the requirement of scientists, the error-bounded lossy compression needs to constrain the point-wise compression errors within a certain value, which differs from compression techniques for traditional data such as JPEG-2000 \cite{taubman2002jpeg2000} for image data and h.265 \cite{sullivan2012overview} for video data. The error-bounded scientific compressors are classified into four main categories: prediction-based, transform-based, dimension-reduction-based, and neural-network-based. They also essentially utilize approaches to manage the data distortion in line with user-specified error bounds.

The prediction-based compressors use data prediction techniques, like linear regression \cite{Xin-bigdata18} and dynamic spline interpolations \cite{szinterp}, to anticipate the data points. Well-known examples are SZ2 \cite{Xin-bigdata18} and SZ3 \cite{sz3,szinterp}. Transform-based compressors, on the other hand, use data transformations to de-correlate the data, then switch to compress the more compressible transformed coefficients. ZFP \cite{zfp}, for example, is a typical example that employs exponent alignment, orthogonal discrete transform, and embedded encoding. SPERR \cite{SPERR}, a more recent work, leverages wavelet transform for data compression. Dimension-reduction-based compressors apply dimension reduction techniques, with (high-order) singular vector decomposition (SVD) being a case in point (for instance, TTHRESH \cite{ballester2019tthresh}). Neural-network-based compressors \cite{ae-sz,glaws2020deep,liu2021high,hayne2021using} utilize neural network models like the autoencoder family \cite{ae,vae,swae}, however, the speeds of them and relatively quite slow.

The aforementioned compressors each have their strengths and weaknesses, depending on the nature of the input data and user needs. To enhance scientific error-bounded lossy compression, two emerging approaches are raised to further refine the specialization of the compressor or to boost its versatility. Regarding compressor specialization, MDZ \cite{mdz}, a prediction-based compressor, is specifically tailored for molecular dynamics simulation data. SZx \cite{szx} offers low-ratio lossy compression at incredibly high speeds. CuSZ \cite{cusz}, CuSZ+ \cite{cusz+} and FZ-GPU \cite{FZGPU} delve into GPU-based scientific lossy compression to quicken the compression process. \cite{jiao2022toward} aims at maintaining the quantities of interest (QoI) of the input data. When it comes to enhancing the versatility of lossy compressors, QoZ \cite{qoz} integrates user-specified quality metric optimization targets and anchor-point-based level-wise interpolation auto-tuning into the SZ3 compression framework. This can effectively improve the compression quality with limited speed degradation. FAZ \cite{liu2023faz}, a hybrid compression framework, combines diverse compression techniques and adaptively generates the compression pipeline for varying inputs, while suffering from low compression speed.

With all those evolving works taken into insight, there is still a lack of broad-spectrum scientific error-bounded lossy compressors that can achieve both top-tier compression quality and adequate compression speed. In this paper, our proposed solution endeavors to fill this gap: we pursue both high compression quality (by optimizing the rate-distortion) and high execution throughput across a wide range of scientific datasets. 

\section{Problem Formulation and Analysis}
\label{sec:backpro}

In this section, we mathematically formulate our research target and then present the fundamental analysis for addressing the target. With those analyses, we can determine the best-fit archetype for the to-be-proposed compressor HPEZ.  

\subsection{Problem Formulation}
The target of HPEZ is to jointly optimize the compression ratio and the user-specified quality metrics (PSNR, SSIM, etc.). Moreover, the proposed new compressor is expected to have relatively high compression and decompression speeds and be well-adapted to diverse types of input data (integer and floating point, single-dimensional and multi-dimensional, and so on).

Eq. \ref{eq:problem} is the formulated research target in this paper. A compressor $C$ and a decompressor $D$ compose the error-bounded lossy compression framework, together with their configuration parameters (denoted by $\theta$). With the input data (denoted by $X$) and a user-specified absolute error bound $e$, the compression framework generates compressed data (denoted as $Z=C_\theta(X)$) and the decompressed data (denoted as $X^{'}=D_\theta(Z)$), which should strictly respect the error bound (denoted $e$) point-wisely. Under those mandatory conditions, HPEZ determines $C$, $D$, and $\theta$ by optimizing the compression ratio under a user-specified quality metric requirement (denoted as $m_0$). Each quality metric corresponds to (and is calculated from) a function $M$, which can be chosen from PSNR, SSIM, a constant function (in case no quality but just compression ratio is concerned), etc. 
Moreover, to ensure the applicability of our proposed compressor for various use cases, we would like the proposed compressor to become a high-performance compressor (including SZ3, QoZ, et al.) having an overall execution speed of at least comparable to SZ3.

\begin{equation}
\label{eq:problem}
\begin{split}
&C,D,\theta=\underset{C,D,\theta}{\arg\max} \hspace{1mm}\frac{|X|}{|Z|} \\
s.t. \ &|x_i-x_i^{'}|\leq e, \forall x_i \in X \\
\ & M(X,X^{'})=m_0 \\
\end{split}
\end{equation}

\subsection{Determining the Best-fit Compressor Archetype for HPEZ}
\label{sec:analysis}
As mentioned before, our proposed compressor should exhibit both good rate-distortion and relatively high speeds. To this end, we need to investigate existing scientific error-bounded lossy compressors to identify the best-fit compressor archetype for our design. The categorization of compressors is priorly discussed in Section \ref{sec:introduction} and Section \ref{sec:related}, but to conduct a deeper analysis here we categorize the existing compressors into more types according to their designs:

\begin{itemize}
    \item Hybrid-data-prediction-based: Applying multiple data predictors for data prediction and reconstruction, such as regressors and Lorenzo predictors \cite{Xin-bigdata18,sz-auto}.
    \item Interpolation-based: Leveraging interpolations for prediction-based data compression \cite{szinterp,qoz}.
    \item Discrete-orthogonal-transform-based: Making use of block-wise Discrete Orthogonal Transform and embedded coding in the compression \cite{zfp}.
    \item Wavelet-transform-based: Combining wavelet transforms and coefficient encoding methods for compression \cite{SPERR,liu2023faz}.
    \item SVD-based: In TTHRESH \cite{ballester2019tthresh}, high-order singular value decomposition is the core of its data processing techniques.
    \item Deep-learning-based: Quite a few deep-learning-based error-bounded lossy compressors have been proposed. Among them, there are autoencoder-based ones \cite{ae-sz,hayne2021using} and \newline coordinate-network-based ones \cite{han2022coordnet,lu2021compressive}.
\end{itemize}

Several existing works \cite{qoz,liu2023faz,ae-sz} have also conducted systematic and thorough experimental analyses of those compressors in diverse types, having tested them in multiple aspects including and not limited to execution speeds, rate-distortion, and practical use cases (e.g. I/O throughput). We conclude their findings as follows:
\begin{itemize}
    \item Despite their great potential in achieving high compression ratios, wavelets-based and SVD-based compressors suffer from low compression speeds due to high computational costs. With fixed data processing strategies, certain examples of them such as SPERR and TTHRESH also fail to perform well in terms of rate-distortion on some data inputs.
    \item Discrete-orthogonal-transform-based ZFP has a very high compression efficiency, but it only presents quite limited compression ratios.
    \item The practicality of current deep-learning-based compressors is also not satisfactory. The networks integrated into them either need per-data online training (for each compression task) or large sizes of training data from the same application for pre-training. This fact greatly damages the availability and efficiency of deep-learning-based compressors.
    \item Compared with others, prediction-based compressors (including hybrid-data-prediction-based and interpolation-based ones) have the advantage of achieving both good compression ratios and acceptable compression speeds. Among them, interpolation-based compressors such as SZ3 \cite{sz3} and QoZ \cite{qoz} optimize the compression rate-distortion. In the experiments carried out by \cite{qoz}, QoZ shows the best performance in the parallel I/O throughput tests.
\end{itemize}

According to our research target and the pros and cons of existing compressor archetypes, we develop a novel high-performance effective compressor namely HPEZ based on the interpolation-based compressor design. In Section \ref{sec:framework}, \ref{sec:interp}, and \ref{sec:autotuning}, we will fully demonstrate the design details of HPEZ, including the research background and newly proposed features.

\section{HPEZ Design Overview}
\label{sec:framework}
In this section, we propose an overview of the HPEZ compressor. As an interpolation-based scientific error-bounded lossy compressor, HPEZ is designed for structured data grids in types of floating points and integers. HPEZ is adaptive to either one-dimensional (1D) or multi-dimensional (2D, 3D, 4D ...) inputs, and exploits the dimension-wise spatial correlations and smoothness of them. 
HPEZ also has the potential to be applied to other domains including image and video because those data are also formatted as (or can be transformed into) structured data grids. The compression framework of HPEZ is illustrated in Figure \ref{fig:framework}. HPEZ takes advantage of the SZ3 modular framework \cite{sz3}, which contains the auto-tuning module, data prediction module, error quantization module, Huffman encoding module, and the Zstd lossless module. The detailed demonstration of the HPEZ compression pipeline is as follows:
\begin{itemize}
    \item \textbf{Step 1: Auto-tuning}. With a user-specified quality metric optimization target, HPEZ first auto-tunes its predictor configurations, which will be featured in Section \ref{sec:autotuning}.
    \item \textbf{Step 2: Data prediction}: HPEZ applies the auto-tuned data predictor on the whole input, acquiring the prediction errors.
    \item \textbf{Step 3: Linear quantization (error control)}: A linear error quantization module quantizes the data prediction errors in step 2 to control the element-wise decompression error. For example, for each data value $x$ and its prediction $x^{'}$, the original error is $e=x-x^{'}$ and the quantized error $e_q$ satisfies $|e_q-e|<=\epsilon$ ($\epsilon$ is the error bound). In this way, we can use $x^{'}+e_q$ as the decompression of $x$ which is bounded by $\epsilon$.
    \item \textbf{Step 4: Huffman encoding}: The quantized prediction errors acquired from Step 3 are further encoded with Huffman encoding. A more concentrated distribution of quantization errors will lower the encoded tree size, therefore the reduction of prediction error is key to improving the compression ratio.
    \item \textbf{Step 5: Lossless postprocessing}: The encoded quantized errors and other metadata are losslessly compressed by Zstd \cite{zstd} to further reduce the compressed size.
\end{itemize}
  HPEZ leverages existing modules in stereotype prediction-based error-bounded compression model (orange ones in Figure \ref{fig:framework}) and interpolation techniques (yellow ones in Figure \ref{fig:framework}). Most importantly, our HPEZ framework introduces several new modules and significantly improved components (as marked in blue and pink), including interpolation designs and auto-tuning techniques. In the data prediction module and the auto-tuning module, new designs have been incorporated in HPEZ to enhance the compression rate-distortion substantially. With those new designs, first, we have significantly improved the interpolation-based data predictors in HPEZ, introducing multiple refinements upon the existing dynamic spline interpolation; Second, the auto-tuning module of HPEZ has also been facilitated with new components for handling new interpolation configurations and boosting adaptability for more datasets. Third, the compression speed of HPEZ still maintains at a high level, empowering it to well-fit efficiency-oriented tasks. Those newly proposed designs will be demonstrated in Section \ref{sec:interp} and Section \ref{sec:autotuning}.
\begin{figure}[ht]
  \centering
  \vspace{1mm}
  \raisebox{-1mm}{\includegraphics[scale=0.55]{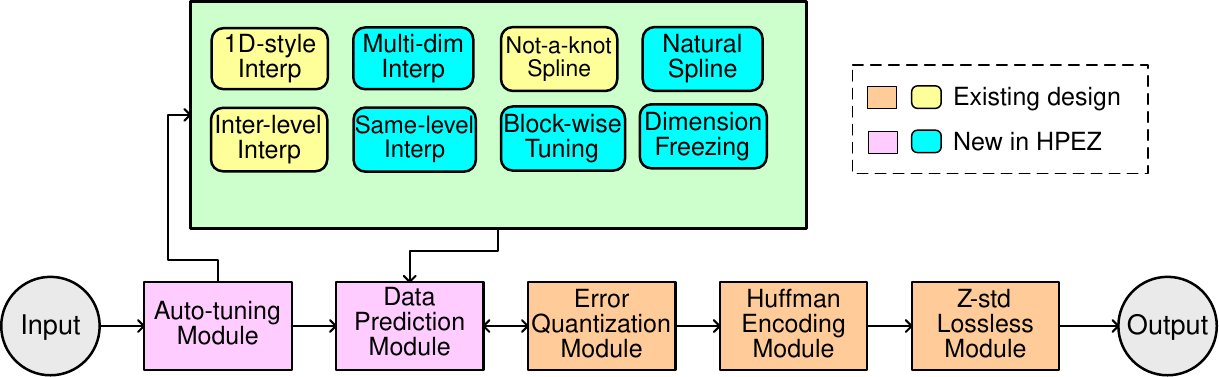}}
  \vspace{-1mm}
  \caption{HPEZ framework}
  \label{fig:framework}
\end{figure}

\section{HPEZ Interpolation-based Predictor}
\label{sec:interp}

In this section, we describe the details of our fine-tuned multi-component interpolation-based data predictor for HPEZ. Compared to the existing interpolation-based predictors, the HPEZ interpolation-based predictor projects a significant improvement over them, attributed to several new components we designed and proposed. These components can obtain a significantly improved prediction accuracy, thus leading to much better rate distortions in the compression. Those new designs together with the existing interpolation designs will be described in the rest of this section and will get auto-tuned for optimization of compression quality (to be detailed in Section \ref{sec:autotuning}). 
\subsection{Overview of Interpolation-based Prediction}
\label{sec:interpoverview}
The interpolation-based data prediction and reconstruction in HPEZ follow the hierarchical anchor-based level-wise dynamic spline interpolation concept, whose prototype was first proposed in SZ3 \cite{szinterp} and then developed in QoZ \cite{qoz}. Figure \ref{fig:interp} presents the interpolation-based data prediction process in the QoZ compressor. Initialized with a sparse losslessly-saved grid, on each interpolation level, the predictor expands the predicted/reconstructed data grid by 2$\times$ (on each dimension), until all data points are predicted/reconstructed. The interpolations with larger strides are performed at higher levels, and the interpolation stride reduces (halved) as the level goes down. We refer the readers to read \cite{qoz} for details. The key features of QoZ level-wize interpolation method include:
\begin{itemize}
    \item Storing anchor points losslessly (with a fixed anchor stride);
    \item The interpolations are done hierarchically (level by level), from large strides (half of the anchor stride) to small strides (1). 
    \item Each level may have different error bounds. Higher levels have smaller error bounds, and the last level always follows the input global error bound. 
    \item Leveraging both linear (first-order) and cubic (third-order) 1-D spline interpolation;
    \item Performing the interpolation along each dimension;
    \item Selecting the best-fit interpolation method for each level; 
    \item Auto-tuning and applying different error-bound values dynamically for different levels; 
\end{itemize}
\begin{figure}[ht]
  \centering
  \raisebox{-1mm}{\includegraphics[scale=0.5]{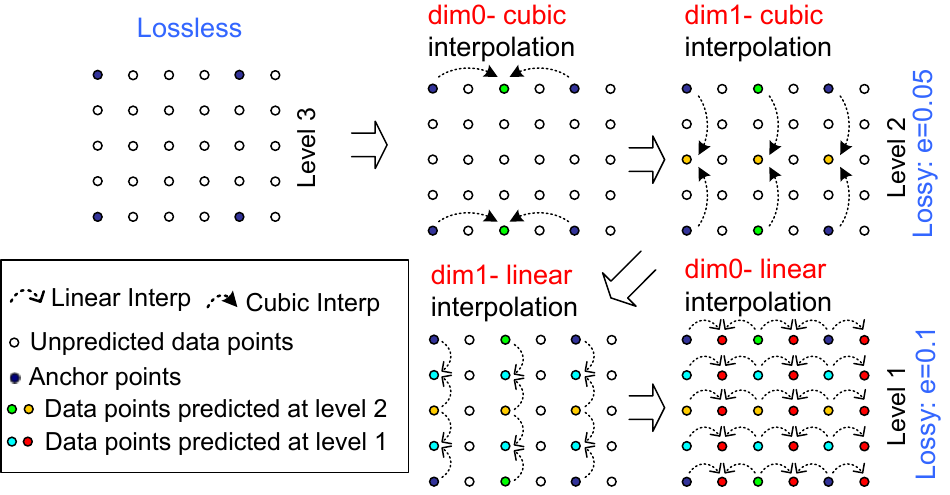}}
  \vspace{-2mm}
  \caption{The anchor-based level-wise dynamic spline interpolation.}
  \label{fig:interp}
\end{figure}

Such an anchor-based level-wise interpolation prediction features three critical advantages. (1) The prediction has a very low time complexity: $O$($N$), where $N$ is the total number of data points in the input dataset.
This is because, for the prediction of each data point, the interpolation is executed just once with an upper-bounded number of neighbor points (e.g. for SZ3/QoZ the upper-bound is 4), and the quantization of its prediction error is also completed in constant time. (2) The level-wise design allows it to set various error bounds at different levels to minimize the negative impact of data compression errors in the data prediction. 
(3) The design of anchor points avoids inaccurate large-stride interpolations, maintaining its prediction accuracy at a relatively high level.

Although the interpolation-based prediction in HPEZ is built upon QoZ, HPEZ proposes several key improvements that significantly boost its prediction accuracy over QoZ, including:
\begin{itemize}
    \item The natural cubic spline function;
    \item The multi-dimensional spline interpolation;
    \item Re-ordering of the interpolations.
\end{itemize}
Next, we will take a deep insight into the interpolation-based prediction in HPEZ, thoroughly demonstrating both the backgrounds and the new characteristics.

\subsection{Spline Interpolation Formulas}
\label{sec:natspline}
All interpolations in HPEZ are based on certain spline interpolation formulas, which interpolate each data point with its neighbors along one dimension. As mentioned in Section \ref{sec:interpoverview}, the spline interpolation formulas are categorized into linear spline interpolation and cubic spline interpolation. Illustrated in Figure \ref{fig:illu of 1d cubic spline}, the data value $d_i$ on index $i$ is going to be predicted by a prediction $p_i$ with the known data points $d_{i-3}$, $d_{i-1}$, $d_{i+1}$, and $d_{i+3}$ in its neighbours. The linear spline interpolation just applies 2 of them with the following formula:
\begin{equation}
\label{eq:lininterp}
\begin{array}{l}
p_i=\frac{1}{2}d_{i-1}+\frac{1}{2}d_{i+1}
\end{array}
\end{equation}

The cubic spline interpolation formulas leverage all the 4 neighbor points, and the formulas are deducted from 3 cubic spline functions ($f_{1}(x)$, $f_{2}(x)$, and $f_{3}(x)$):
\begin{figure}[ht]
  \centering
  \raisebox{-1cm}{\includegraphics[scale=0.6]{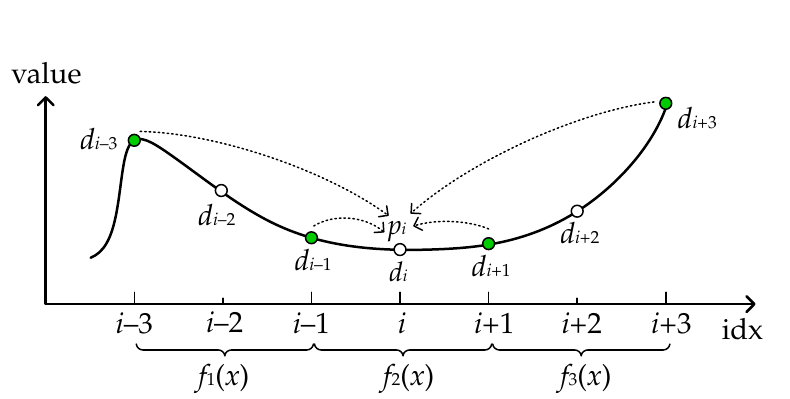}}
  \caption{Illustration of 1D cubic spline interpolation.}
  \label{fig:illu of 1d cubic spline}
  \vspace{-2mm}
\end{figure}
\begin{equation}
\label{eq: cubic function}
\begin{array}{l}
\hspace{-3mm}f_{1}(x)
=a_1(x\hspace{-1mm}-\hspace{-1mm}(i\hspace{-1mm}-\hspace{-1mm}3))^3
\hspace{-1mm}+\hspace{-1mm}b_1(x\hspace{-1mm}-\hspace{-1mm}(i\hspace{-1mm}-\hspace{-1mm}3))^2
\hspace{-1mm}+\hspace{-1mm}c_1(x\hspace{-1mm}-\hspace{-1mm}(i\hspace{-1mm}-\hspace{-1mm}3))
\hspace{-1mm}+\hspace{-1mm}\delta_1 \hspace{-1mm} \\

 \hspace{-3mm}f_{2}(x)
 =a_2(x\hspace{-1mm}-\hspace{-1mm}(i\hspace{-1mm}-\hspace{-1mm}1))^3
\hspace{-1mm} +\hspace{-1mm}b_2(x\hspace{-1mm}-\hspace{-1mm}(i\hspace{-1mm}-\hspace{-1mm}1))^2
 \hspace{-1mm}+\hspace{-1mm}c_2(x\hspace{-1mm}-\hspace{-1mm}(i\hspace{-1mm}-\hspace{-1mm}1))
 \hspace{-1mm}+\hspace{-1mm}\delta_2 \hspace{-1mm}\\
 
\hspace{-3mm}f_{3}(x)
=a_3(x\hspace{-1mm}-\hspace{-1mm}(i\hspace{-1mm}+\hspace{-1mm}1))^3
\hspace{-1mm}+\hspace{-1mm}b_3(x\hspace{-1mm}-\hspace{-1mm}(i\hspace{-1mm}+\hspace{-1mm}1))^2 
\hspace{-1mm}+\hspace{-1mm} c_3(x\hspace{-1mm}-\hspace{-1mm}(i\hspace{-1mm}+\hspace{-1mm}1))
\hspace{-1mm}+\hspace{-1mm}\delta_3 \hspace{-1mm}
\end{array}
\end{equation}

The spline functions $f_1$, $f_2$, and $f_3$ have scopes of [$i-$3,$i-$1], [$i-$1,$i+$1], and [$i+$1,$i+$3], respectively. The zero-order, first-order, and second-order interpolation conditions are shown as follows:

\vspace{-3mm}
\begin{equation}
\begin{array}{l}
\label{eq: spline system}
 f_{1}(i-3)= d_{i-3}; \     
 f_{1}(i-1)= d_{i-1} \\
 f_{2}(i-1)= d_{i-1}; \
 f_{2}(i+1)= d_{i+1} \\
 f_3(i+1)=d_{i+1}; \
 f_3(i+3)=d_{i+3}\\
f_{1}^{'}(i-1)=f_{2}^{'}(i-1); \
 f_{2}^{'}(i+1)=f_{3}^{'}(i+1)\\
 f_{1}^{''}(i-1)=f_{2}^{''}(i-1); \ 
 f_{2}^{''}(i+1)=f_{3}^{''}(i+1)\\
\end{array}
\end{equation}

Since $f_1$, $f_2$, and $f_3$ have 12 coefficients in total and Eq. \ref{eq: spline system} only has 10 conditions, two more boundary conditions are needed. The traditional SZ3 and QoZ cubic spline interpolation \cite{szinterp,qoz} applies the following 'not-a-knot' conditions:
\begin{equation}
\vspace{-1mm}
\label{eq:not-a-knot spline boundary}
\begin{array}{l}
 f_{1}^{'''}(i-1)=f_{2}^{'''}(i-1); \
 f_{2}^{'''}(i+1)=f_{3}^{'''}(i+1)\\
\end{array}
\end{equation}

Then with Eq. \ref{eq: spline system} and Eq. \ref{eq:not-a-knot spline boundary}, the prediction value of $p_i$ is:
\begin{equation}
\vspace{-1mm}
\label{eq: not-a-knot spline}
\begin{array}{l}
p_i=f_2(i)=-\frac{1}{16}d_{i-3} + \frac{9}{16}d_{i-1} + \frac{9}{16}d_{i+1} -\frac{1}{16}d_{i+3}
\end{array}
\end{equation}

However, there are other choices for the 2 boundary conditions, which may lead to different cubic spline interpolation formulas. We explore another set of boundary conditions: the natural spline condition, which is:
 \begin{equation}
 \vspace{-2mm}
 \label{eq:nature spline boundary}
 \begin{array}{l}
  f_{1}^{''}(i-3)=0; 
  f_{3}^{''}(i+3)=0\\
 \end{array}
 \vspace{1mm}
 \end{equation}

Combining Eq. \ref{eq: spline system} and Eq. \ref{eq:nature spline boundary}, the interpolation function for predicting $p_i$ would be written as:
 \begin{equation}
 \label{eq: natural spline}
 \begin{array}{l}
 p_i=f_2(i)=-\frac{3}{40}d_{i-3} + \frac{23}{40}d_{i-1} + \frac{23}{40}d_{i+1} -\frac{3}{40}d_{i+3}
 \end{array}
 \end{equation}

Our experiments with multiple datasets under diverse error thresholds showed that Eq. \ref{eq:lininterp}, Eq. \ref{eq: not-a-knot spline}, and Eq. \ref{eq: natural spline} have distinct advantages. In different cases, each of them is able to outperform others. Therefore, we employ all 3 of them and dynamically select from them for each task.
\subsection{1D and Multi-dimensional Spline Interpolation}
\label{sec:mdinterp}
In traditional interpolation-based compressors, for each data point, the interpolation is performed along a single dimension, so we need to switch the interpolation directions during this process and arrange an order for those directions. In the following text, we call the interpolation method adopted by SZ3/QoZ \textit{1D-style interpolation}. As an example, in Figure \ref{fig:1dvsmd} (a), the 1D-style interpolation first proceeds interpolations along Dim0, then performs the rest of the interpolations along Dim1. 

Actually, The existing 1D-style interpolation has not fully exploited the multi-dimensional continuity and smoothness of input data arrays, because all the interpolations are constricted in a single-dimensional direction. To address this limitation, we propose a new interpolation paradigm for HPEZ called multi-dimensional spline interpolation, which can take better advantage of data correlation across multiple dimensions. As shown in Figure \ref{fig:1dvsmd} (b), the multi-dimensional spline interpolation initially performs the 1D interpolations for some data points as there are only 1D neighbors at the moment, then it performs 2D interpolations for the remaining data points that already have neighbors in two dimensions. 
The multi-dimensional spline interpolation is symmetric across all the dimensions, meaning that it does not need a selection of dimensional order. 
\begin{figure}[ht]
  \centering
  \raisebox{-1mm}{\includegraphics[scale=0.6]{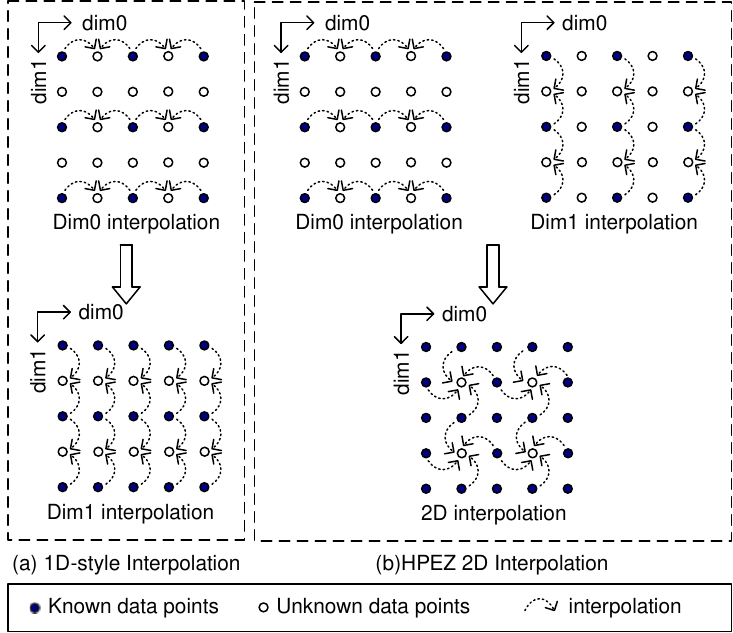}}
  \vspace{-2mm}
  \caption{Comparison of 1D-style interpolation and HPEZ multi-dimensional interpolation (an 2D example).}
  \label{fig:1dvsmd}
\end{figure}

With the main concept of the HPEZ multi-dimensional spline interpolation in mind, two questions remain: how should we carry out the multi-dimensional interpolations specifically, and why does it outperform the 1D-style interpolations?

We feature the HPEZ multi-dimensional interpolation as follows. For each data point $x$, suppose $X_i \ (1 \le i \le n)$ are all the available 1D interpolation results for predicting $x$ (which can either be linear interpolation or cubic interpolation and are along all dimensions), the multi-dimensional interpolation result $X^{'}$ is a linear-combination of $X_i$:
 \begin{equation}
 \vspace{-2mm}
 \label{eq: mdinterp}
 X^{'}=\sum_{i=1}^{n}\alpha_iX_i \ \ \ (\sum_{i=1}^{n}\alpha_i=1)
 \end{equation}
\begin{theorem}
\label{theo:md}
With fine-tuned $\alpha_i$, $X^{'}$ would have a no higher prediction error than that of the 1D-style interpolation $X_i$. 
\end{theorem}
\begin{proof}
Without loss of generality, we can regard $\{X_i\}$ and $X^{'}$ as random variables, in which $\{X_i\}$ are independent with each other. When dealing with smooth data inputs, the $\{X_i\}$ can be thought of as no-biased estimations of $x$, i.e. $E(X_i)=x$.

Now consider the $X^{'}$. Since $\sum_{i=1}^{n}\alpha_i=1$, it is easy to know that $E(X^{'})=x$, so $X^{'}$ is still a non-biased estimation of $x$. Because ${X_i}$ are independent with each other, $(X^{'}-x)=\sum_{i=1}^{n}\alpha_i(X_i-x)$ follows the distribution of $ N(0,\sigma^2)$, in which:
\vspace{-1mm}
\begin{equation}
\label{eq:sigma}
\sigma^2=\sum_{i=1}^{n}\alpha_i^2\sigma_i^2
\end{equation}
With the Lagrange method, based on the constraint $\sum_{i=1}^{n}\alpha_i=1$, 
\begin{equation}
\label{eq:minG}
\begin{split}
\min \sigma^2 & =\frac{\prod_{i=1}^{n}\sigma_i^2}{\sum_{i=1}^{n}\pi_i}
 \le \min \{\sigma_1^2,\sigma_2^2,...\sigma_n^2\} \
 ( \pi_i=\frac{\prod_{j=1}^{n}\sigma_j^2}{\sigma_i^2} )
\end{split}
\end{equation}
, and the minimum is obtained when:
\begin{equation}
\label{eq:alpha}
\begin{split}
\alpha_i^*=\frac{\pi_i}{\sum_{j=1}^{n}\pi_j}
\end{split}
\end{equation}

As such, we have proved that, if the $\{\alpha_i\}$ is selected based on Eq. \ref{eq:alpha}, the prediction error variance of the multi-dimensional interpolation $X^{'}$ will be no larger than each of the 1D-style interpolation $X_i$ according to Eq. \ref{eq:minG}. So, the average L-1 prediction error will also be minimized.
\end{proof}

How to determine $\alpha_i^*$ in HPEZ (i.e. how to estimate $\sigma_i^{2}$) will be detailed in Section \ref{sec:autotuning}. 

\subsection{Interpolation Re-ordering}
\label{sec:reorder}
After the proposal of natural cubic spline and multi-dimensional interpolation, HPEZ also introduces interpolation re-ordering, which improves both prediction accuracy and prediction speed. It includes two aspects: the fast-varying-first interpolation and same-level cubic interpolation.
\subsubsection{Fast-varying-first interpolation}
\label{sec:fvfi}
In the existing implementation of 1D interpolations, the interpolations are executed axis by axis on the input dataset, and along each axis, the interpolations are performed 'slice by slice'. The 'slice' here means a slice of the data array along an interpolation axis. Figure \ref{fig:fvf-2d} (a) presents a 2D example for the order of interpolations adopted by QoZ (and also SZ3): the interpolations are performed in the sequence of numbers (\circled{1}, \circled{2}, \circled{3}, $\cdots$). For the interpolation along Dim0 in QoZ, it follows dim0-major order: the interpolation is executed along Dim0 with a higher preference compared with Dim1. However, when Dim1 is the fastest-varying-dimension 
, this interpolation order may fall into a bad cache usage because it is successively accessing data points located distantly in the memory. 
To resolve this issue, HPEZ re-arranges the interpolation order, having the interpolations first move along the fast-varying dimension (the Dim1-major style as in Figure \ref{fig:fvf-2d}), as demonstrated in Figure \ref{fig:fvf-2d} (b). The interpolation position first traverses through Dim1 and then moves along Dim0. In this way, the data points are accessed sequentially with shorter distances in the memory so that the cache usage can be optimized, greatly saving the memory access cost. 

\begin{figure}[ht]
  \centering
  \raisebox{-1mm}{\includegraphics[scale=0.5]{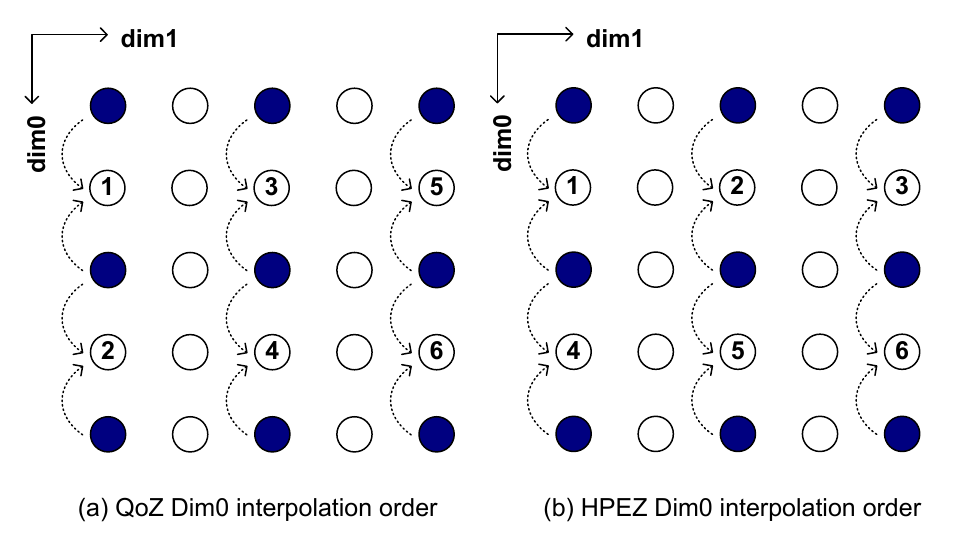}}
  \vspace{-2mm}
  \caption{Comparison of QoZ and HPEZ interpolation orders (Dim1 is the fastest-varying dimension)}
  \label{fig:fvf-2d}
\end{figure}
\begin{figure}[ht]
  \centering
  \raisebox{-1mm}{\includegraphics[scale=0.48]{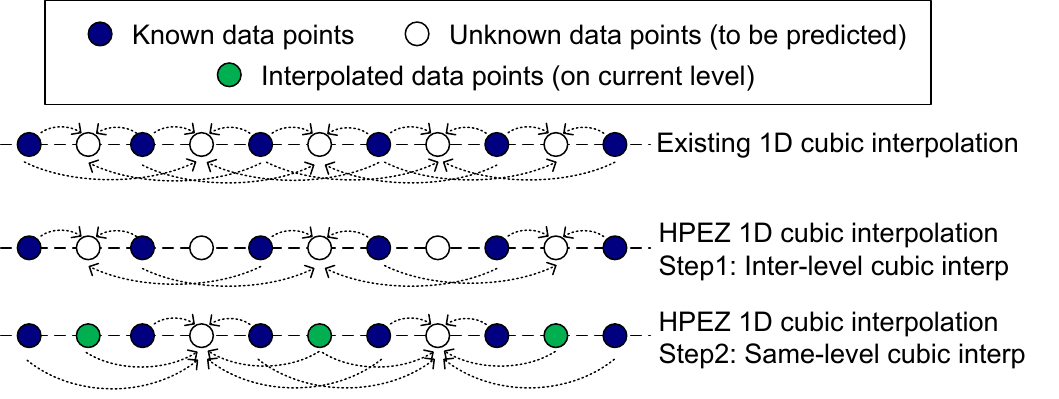}}
  \vspace{-1mm}
  \caption{Illustration of same-level cubic interpolation}
  \label{fig:iii}
\end{figure}
\subsubsection{Same-level cubic interpolation}

We develop a new same-level cubic interpolation in HPEZ, which can further improve prediction accuracy.
In the traditional interpolation design \cite{qoz,szinterp}, at each interpolation level, the neighbor points of each data point to be interpolated are limited on the higher levels (interpolation levels with larger strides). For the 1D cubic spline interpolation applied on a data point with stride $s$, 4 neighbor points with distance $s$ and $3s$ are used, which have been predicted on the higher interpolation levels. As shown in Figure \ref{fig:iii} (note that $s$ is the distance between each closest hollow and solid point), the first row shows this interpolation method, in which all the hollow data points (on the current interpolation level) are predicted by the solid data points (on higher interpolation levels). If we are able to include more neighbors for each point (for example, the 2 white points with a distance of $2s$ to it), the prediction accuracy can be improved. As illustrated in the 2nd and 3rd rows of Figure \ref{fig:iii}, instead of traversing through all the white data points in one step, HPEZ splits the 1D cubic spline interpolation into 2 steps. In the first step (the second row of Figure \ref{fig:iii}), half of the white points are interpolated by inter-level interpolation (the existing interpolation) with 4 neighbor points. 
In the second round, the rest half of the white points are interpolated by the same-level interpolation with 6 neighbor points for each, 
including points interpolated on higher interpolation levels and the current interpolation level. With this new interpolation, half of the data points are predicted with two more neighbor points to achieve better prediction accuracy. Similar to the deductions in Section \ref{sec:natspline}, for a data point $p_i$, with its 6 neighbor points $d_{i-3}$, $d_{i-2}$, $d_{i-1}$, $d_{i+1}$, $d_{i+2}$, and $d_{i+3}$ the same-level cubic spline interpolation formula would be the following two. Eq. \ref{eq: not-a-knot inlv spline} is for the not-a-knot cubic spline and Eq. \ref{eq: natural inlv spline} is for the natural cubic spline. The same strategy can also be extended to the multi-dimensional interpolation, splitting it into 2 steps each with halved data points. 
\begin{equation}
\vspace{-2mm}
\label{eq: not-a-knot inlv spline}
\begin{array}{l}
p_i=-\frac{1}{6}d_{i-2} + \frac{4}{6}d_{i-1} + \frac{4}{6}d_{i+1} -\frac{1}{6}d_{i+2}
\end{array}
\vspace{1.5mm}
\end{equation}
\begin{equation}
\label{eq: natural inlv spline}
\begin{array}{l}
\hspace{-1mm}p_i \hspace{-0.5mm}=\hspace{-0.5mm}\frac{3}{62}d_{i-3}\hspace{-0.2mm}-\hspace{-0.2mm}\frac{18}{62}d_{i-2} \hspace{-0.2mm}+\hspace{-0.2mm} \frac{46}{62}d_{i-1} \hspace{-0.2mm}+\hspace{-0.2mm} \frac{46}{62}d_{i+1} \hspace{-0.2mm}-\hspace{-0.2mm}\frac{18}{62}d_{i+2}\hspace{-0.2mm}+\hspace{-0.2mm}\frac{3}{62}d_{i+3}

\end{array}
\end{equation}

\section{HPEZ Auto-tuning Modules}
\label{sec:autotuning}
we developed an advanced auto-tuning module in HPEZ, which plays a critical role in preserving and optimizing the compression quality by making the best use of the abundant interpolation options offered by HPEZ which are discussed in Section \ref{sec:interp}. 
Figure \ref{fig:autotuning} displays all the components and processes of the HPEZ auto-tuning module. This module inherits the interpolation error-bound tuning process from QoZ \cite{qoz}, while substantially upgrading the QoZ 'global' interpolation tuning process. Specifically, HPEZ exploits several brand-new processes: dynamic dimension freezing tuning, block-wise interpolation tuning, Lorenzo tuning, and a data sampling/statistical analysis process supporting those tuning processes. In the remainder of this section, we present the detailed design of the auto-tuning-related components in HPEZ.
\begin{figure}[ht]
  \centering
  \vspace{1mm}
  \raisebox{-1mm}{\includegraphics[scale=0.56]{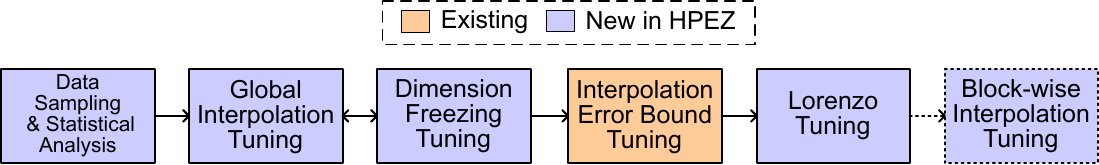}}
  \vspace{-1mm}
  \caption{HPEZ auto-tuning module}
  \label{fig:autotuning}
\end{figure}
\vspace{-2mm}
\subsection{Data Sampling and Statistical Analysis}
\label{sec:stats}
The data sampling and statistical analysis is an auxiliary process of the HPEZ auto-tuning module. In this process, HPEZ uniformly samples a small portion from the full data input (based on a hyper-parameter with the default sampling rate of 0.2\%), and then it performs the 1D interpolation (both linear and cubic) on those data points with their neighbors along all dimensions. Afterward, the mean square errors (MSE) of the interpolations along different dimensions can serve as the estimations of the interpolation error variances ($\sigma_i^{2}$) described in Section \ref{sec:mdinterp}. Thus, it can be used to determine the most non-smooth dimension in the data for dynamic dimension freezing (Section \ref{sec:freeze}) by selecting the dimension with the largest interpolation MSE.

\subsection{Global Interpolation Tuning}
The global interpolation tuning process in HPEZ is derived from the predictor tuning process proposed in QoZ, which aims to select the best-fit interpolation configuration from different choices 
Specifically, at each interpolation level, the global interpolation tuning process makes the following selection for the input data:

\begin{itemize}
    \item \textbf{Existing in QoZ}: The order of interpolation (linear or cubic);
    \item \textbf{Existing in QoZ}: The dimensional order (only for 1D-style interpolation);
    \item \textbf{New in HPEZ}: The type of cubic spline (not-a-knot or natural, only for cubic interpolation);
    \item \textbf{New in HPEZ}: The interpolation paradigm (1D-style or multi-dimensional);
    
    \item \textbf{New in HPEZ}: Applying inner-level interpolation or not (only for cubic interpolation);
\end{itemize}

Similar to QoZ, the sampled data are used for performing compression tests with all the available interpolation configurations. Then, HPEZ selects the interpolation configuration with the lowest average absolute prediction error as the final tuning result.

\subsection{Dynamic Dimension Freezing}
\label{sec:freeze}
The dynamic dimension freezing in HPEZ is designed to avoid inaccurate interpolation predictions along non-smooth dimensions. 
For a multi-dimensional input data array, it may present fine smoothness along some of its dimensions but present bad smoothness along the other dimensions. 
In those cases, both the 1D-style and multi-dimensional interpolation will fail in achieving high prediction accuracy as they will involve interpolations along non-smooth directions.
The dimension freezing is that, given one dimension, HPEZ sets anchor points along those dimensions with stride 1 (without intervals) and never performs interpolations along those dimensions. Figure \ref{fig:freeze} uses the interpolation on a 3D data block as an example of dimension freezing. For a clear view, only the 1D interpolations are shown. Figure \ref{fig:freeze} (a) is the normal 1D interpolations without a frozen dimension, and Figure \ref{fig:freeze} (b) is the 1D interpolations with a dimension frozen, in which no interpolations are made along the frozen dimension. With this dynamic strategy, HPEZ does not require data smoothness along all dimensions to optimize its compression ratio.
According to our experimental results, compared to the highly improved prediction accuracy and greatly reduced quantization bin size, the storage overhead for additional anchor points is affordable. To determine whether to freeze a dimension and which dimension should be frozen, the auto-tuning module of HPEZ first specifies the most non-smooth dimension in the input data array in the statistical analysis (Section \ref{sec:stats}), then separately tunes 2 optimized interpolation configurations with/without this dimension frozen. If freezing this dimension presents a better compression ratio, HPEZ will freeze this dimension.

\begin{figure}[ht]
  \centering
  \raisebox{-1mm}{\includegraphics[scale=0.5]{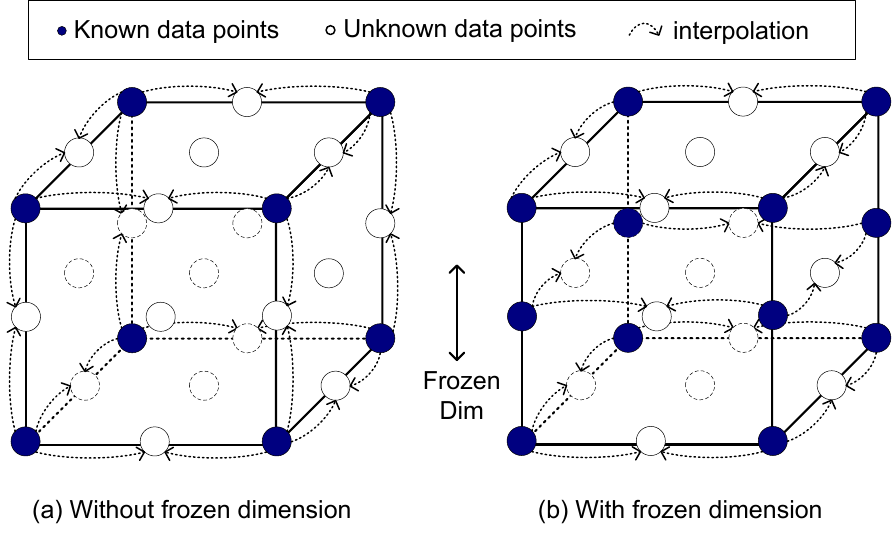}}
  \vspace{-2mm}
  \caption{Illustration of dimension freeze}
  \label{fig:freeze}
\end{figure}
\vspace{-2mm}

\subsection{Interpolation Error Bound Tuning}
\label{sec:ebtuning}
Previously indicated in Section \ref{sec:interpoverview}, the HPEZ interpolations on each interpolation level follow a separate dynamically auto-tuned error bound. 
For the level-wise error-bound setting, HPEZ follows the same design as in QoZ \cite{qoz}. The error bound for each level is computed by Eq. \ref{eq:ebl}, in which $\alpha$ and $\beta$ are tunable parameters:
\vspace{-2mm}
\begin{equation}
\label{eq:ebl}
e_l=\frac{e}{min(\alpha^{l-1},\beta)} \ (\alpha \geq 1 \ and \ \beta \geq 1)
\vspace{-2mm}
\end{equation}

In the auto-tuning process for determining $\alpha$ and $\beta$, HPEZ also leverages the module proposed in QoZ. We refer the readers to check \cite{qoz} for details.
\vspace{-2mm}
\subsection{Tuning with Lorenzo Predictor}
\label{sec:lorenzo}
Leveraged in SZ3 but excluded by QoZ, the dynamic-order Lorenzo predictor designed in \cite{sz-auto} is involved in HPEZ, as it is still an essential supplement of interpolation-based predictors for high-accuracy low-compression-ratio cases \cite{szinterp,sz3, liu2023faz}. In the auto-tuning compression test process, after the auto-tuning module has acquired the optimized interpolation-based rate-distortion pair and its corresponding configuration, the auto-tuning module runs one more compression test with the Lorenzo predictor, then makes the selection between the interpolation-based predictor and the Lorenzo predictor according to the pre-given optimization target. Following the design in \cite{liu2023faz}, a multiplicative coefficient is applied to adjust the bit rate estimation of the Lorenzo predictor.

\vspace{-1mm}
\subsection{Block-wise Interpolation Tuning}
\label{sec:blockwise}
If the interpolation predictor is finally selected after the Lorenzo tuning, the block-wise interpolation tuning will fine-tune the interpolation configuration separately on each data block. 
Various regions of the input data will exhibit different characteristics (such as dimension-wise smoothness), which makes them adapt to different interpolation configurations accordingly.
To address this issue, HPEZ introduces the block-wise interpolation tuning process into its auto-tuning module, dedicated to identifying the best-fit interpolation configurations for diverse segments of the data. Figure \ref{fig:blockwise} shows the details of the HPEZ block-wise interpolation tuning. First, after the auto-tuning has globally determined the level-wise interpolation error bounds (Figure \ref{fig:blockwise} (a)), the input data array is split into blocks (Figure \ref{fig:blockwise} (b)) of the same size. On each data block, a sub-block (in default has 4\% of the full block size) is sampled out in the center of this block (Figure \ref{fig:blockwise} (c)), and then the interpolation configuration for this block (Figure \ref{fig:blockwise} (d)) is tuned by the compression tests performed on the sampled sub-block. The block size for block-wise interpolation tuning is a hyper-parameter in HPEZ, and after primary experiments, we use the default value of 32 for it.
\begin{figure}[ht]
  \centering
  \vspace{2mm}
  \raisebox{-1mm}{\includegraphics[scale=0.35]{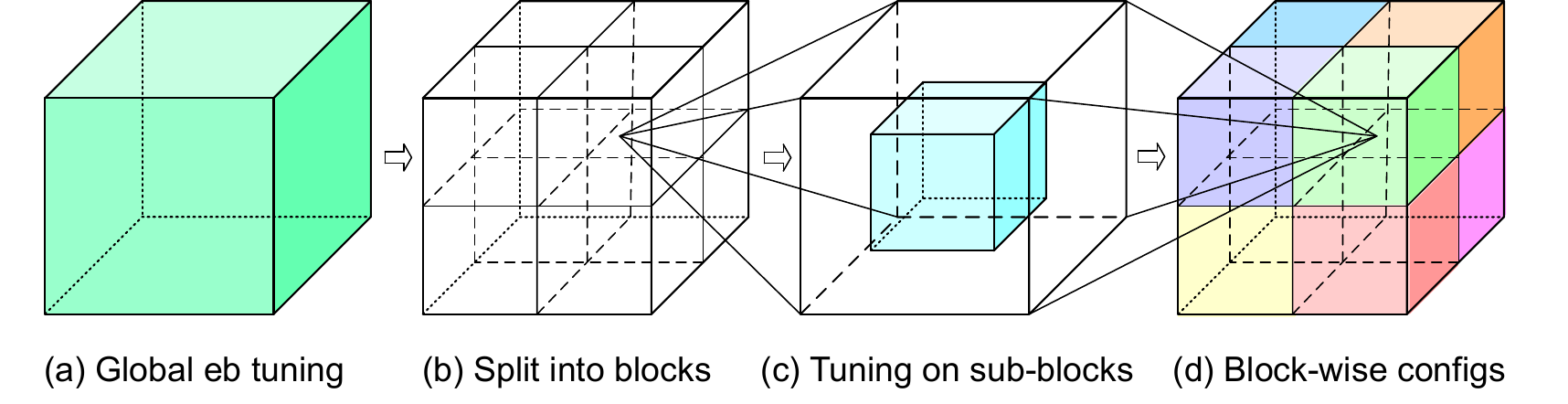}}
  \vspace{-2mm}
  \caption{Block-wise interpolation tuning}
  \label{fig:blockwise}
\end{figure}

\section{Performance Evaluation}
\label{sec:evaluation}

To verify the effectiveness and efficiency of HPEZ, systematical evaluations of HPEZ together with six other state-of-the-art error-bounded lossy compressors are presented in this section. 

\subsection{Experimental Setup}
\label{sec:setup}

\subsubsection{Experimental environment and datasets} 
\label{sec:env}
We conducted all the evaluation experiments on the Purdue Anvil supercomputer (for all experiments) and the Argonne Bebop supercomputer (for the Globus-based data transfer test). On the Anvil supercomputer, each computing node features two AMD EPYC 7763 CPUs with 64 cores having a 2.45GHz clock rate and 256 GB DDR4-3200 RAM. The computing node we used on the Bebop has the Intel Xeon E5-2695v4 CPU with 64 CPU cores and a total of 128GB of DRAM.

In order to evaluate the compressors more comprehensively and systematically, 8 real-world scientific applications from diverse scientific domains that have been frequently used for the evaluation of scientific data error-bounded lossy compression \cite{sdrb} are involved in the evaluation.
The detailed information of the datasets is in Table~\ref{tab:dataset information}. As suggested by domain scientists, some fields of the datasets listed above are transformed to their logarithmic domain before compression for better visualization. Among those 8 datasets, 6 are in the floating point type and 2 are in the integer type. Because floating point data are the very majority of scientific data and several of the existing scientific compressors only support floating point data, in the following experiments we mainly focus on the 6 floating point datasets and present the evaluations on the integer datasets as verification of HPEZ for its adaptiveness to scientific integer datasets and other integer datasets (natural images and videos).
\begin{table}[ht]
    \centering
    \caption{Information of the datasets in experiments}
\resizebox{0.75\columnwidth}{!}{  
    \begin{tabular}{|c|c|c|c|c|c|}
    \hline
    App.&\# files& Dimensions & Total Size& Domain & Type\\
    \hline
    RTM \cite{geodriveFirstBreak2020} &37&449$\times$449$\times$235&6.5GB&Seismic Wave&Floating points\\
    \hline
    SEGSalt \cite{SEGSalt}&3&1008$\times$1008$\times$352&4.2GB&Geology&Floating points\\
    \hline
    Miranda \cite{miranda}& 7 & 256$\times$384$\times$384& 1GB& Turbulence&Floating points \\
    \hline
    SCALE-LetKF \cite{scale-letkf} & 12 & 98$\times$1200$\times$1200 & 6.4GB&Climate&Floating points\\
    \hline
    CESM-ATM \cite{cesm}& 33 & 26$\times$1800$\times$3600 &17GB& Weather&Floating points\\
    \hline
    JHTDB \cite{jhtdb}& 10 & 512$\times$512$\times$512 &5GB& Turbulence&Floating points\\
    \hline
    NSTX-GPI \cite{nstx-gpi}& 1 & 50000$\times$80$\times$64 &977MB& Fusion&Integer\\
    \hline
    APS & 5 & 1792$\times$2048 &71MB& Material&Integer\\
    \hline
    \end{tabular}}
    \label{tab:dataset information}
\end{table}
\vspace{-3mm}
\subsubsection{Comparison of lossy compressors in evaluation}
In our experiments, we compare HPEZ with six other error-bounded lossy compressors, which have been verified to have good compression quality and/or performance in prior works \cite{liu2023faz,sz3,szinterp,qoz}. The six compressors can be categorized into \textbf{high-performance compressors} and \textbf{high-ratio compressors}. The high-performance compressors have relatively fast compression speeds with moderate compression ratios, including SZ3.1 \cite{sz3}, ZFP 0.5.5 \cite{zfp}, and QoZ 1.1 \cite{qoz}. The high-ratio compressors achieve a high compression ratio/quality with advanced data processing methods, therefore having relatively low compression speeds. They are SPERR 0.6 \cite{SPERR}, FAZ \cite{liu2023faz}, and TTHRESH \cite{ballester2019tthresh}.
HPEZ should be categorized as a high-performance compressor because it exhibits comparable compression speed with modern high-performance compressors.

We didn't involve deep-learning-based compressors due to the following reasons: 1) Coordinate-network-based compressors suffer from extremely low compression speeds which are far from acceptable. 2) Autoencoder-based compressors also have low compression speeds (not comparable with high-performance compressors. For example, AE-SZ has similar speeds with SPERR \cite{ae-sz}). Meanwhile, their compression ratios are lower than SZ3 as validated in \cite{ae-sz}. 

\vspace{-1mm}

\subsubsection{Experimental configurations and evaluation metrics}
In the compression experiments, the error bound mode we adopted is value-range-based error bound (denoted as $\epsilon$) \cite{z-checker}, which is essentially equivalent to the absolute error bound (denoted as $e$), with the relationship of $e$ = $\epsilon \cdot value\_range$. Since the value-range-based error bound can adapt to diverse amplitudes of datasets, 
it has been broadly used in the lossy compression community \cite{Xin-bigdata18,sz3,liang2021mgard+,sz-auto,liu2023faz}.


We perform the evaluation based on the following key metrics: 

\begin{itemize}
    \item {Speeds: Check the compression and decompression speeds of compressors.}
    \item {Compression ratio (CR) under the same error bound: Compression ratio is the metric mostly cared for by the users. Given the input data $X$ and compressed data $Z$, the compression ratio $CR$ is: $CR=\frac{|X|}{|Z|}$ ( $|\ |$ is the size operator).}
    \item \textit{Rate-PSNR plots}: Plot curves for compressors with the bit rate of the compressed data and the decompression data PSNR. 
    \item \textit{Rate-SSIM plots}: Another rate distortion evaluation  plotting bit rate and SSIM \cite{ssim}. 
     \item {Parallel throughput performance with compressors: Simulate and perform parallel data transfer tests on the distributed scientific database on multiple supercomputers. }
    \item {Visualization with the same CR: Comparing the visual qualities of the reconstructed data from different compressors based on the same CR.}
   
\end{itemize}
\vspace{-4mm}
\subsection{Experimental Results}

\subsubsection{Speeds}
To verify our categorization of compressors and examine the compression efficiency of HPEZ, in Table \ref{tab:eva-seq-speed} we present the compression and decompression speeds of 6 comparison compressors and HPEZ (under error bound 1e$-$3, i.e., $10^{-3}$) on the Anvil machine. From the table, we can clearly observe that the high-performance compressors (SZ 3.1, ZFP 0.5.5, and QoZ 1.1) have far better compression speeds than the high-ratio compressors (SPERR, FAZ, and TTHRESH) with the gap of 3$\times$-10$\times$. Having a speed of around 70\% $\sim$ 90\% of QoZ 1.1, HPEZ can definitely be regarded as a high-performance compressor, achieving 2$\times$ $\sim$ 6$\times$ performance improvement over SPERR/FAZ, and 4$\times$ $\sim$ 17$\times$ performance improvement over TTHRESH. This relatively high speed ensures the advantages of HPEZ on efficiency-oriented and high-ratio-preferred compression tasks. Figure \ref{fig:evaluation-eb-speed} presents the error bound-compression speed curves of each compressor on the 6 tested datasets ( the x-axis is the negative log10 of the error bounded and the y-axis is the compression speed). Those plots also prove that HPEZ is much more efficient than the high-ratio compressors (SPERR, FAZ, and TTHRESH) and has close performances to SZ3 and QoZ.

\begin{table}[ht]
\centering
\footnotesize
  \caption {Execution speeds (MB/s per CPU core) with $\epsilon$=1e-3} 
  \vspace{-2mm}
  \label{tab:eva-seq-speed} 
  \begin{adjustbox}{width=0.75\columnwidth}
\begin{tabular}{|c|c|c|c|c|c|c|c|c|}
\hline
\multirow{2}{*}{Type}          & \multirow{2}{*}{Dataset} & \multirow{2}{*}{SZ 3.1}  & \multirow{2}{*}{ZFP 0.5.5}   & \multirow{2}{*}{QoZ 1.1} & \multirow{2}{*}{SPERR 0.6} & \multirow{2}{*}{FAZ} & \multirow{2}{*}{TTHRESH}  & \multirow{2}{*}{HPEZ} \\
                               &                          &  &  & &                   &                      &  &  \\ \hline
\multirow{6}{*}{\rotatebox[origin=c]{90}{Compression}}   & CESM                     & 219 & 331   & 215 & 49                     & 58                   & 10    & 140 \\ \cline{2-9} 
                               & RTM                      & 211 & 412   & 191 & 63                     & 30                   & 18   & 142 \\ \cline{2-9} 
                               & Miranda                  & 163 & 416   & 157 & 35                     & 29                   & 28   & 140 \\ \cline{2-9} 
                               & SCALE                    & 188 & 191   & 182 & 32                     & 61                   & 17   & 129 \\ \cline{2-9} 
                               & JHTDB                    & 140 & 225   & 122 & 33                     & 28                   & 23   & 105  \\ \cline{2-9} 
                               & SegSalt                  & 189 & 645   & 201 & 51                     & 36                   & 13   & 141 \\ \hline
\multirow{6}{*}{\rotatebox[origin=c]{90}{Decompression}} & CESM                     & 661 & 584   & 689 & 92                     & 101                  & 53   & 513 \\ \cline{2-9} 
                               & RTM                      & 786 & 622   & 626 & 124                    & 64                   & 108  & 510 \\ \cline{2-9} 
                               & Miranda                  & 419 & 946   & 351 & 75                     & 60                   & 111  & 473 \\ \cline{2-9} 
                               & SCALE                    & 610 & 553   & 567 & 68                     & 140                  & 53   & 450 \\ \cline{2-9} 
                               & JHTDB                    & 376 & 425   & 243 & 70                     & 59                   & 60   & 330 \\ \cline{2-9} 
                               & SegSalt                  & 592 & 1060  & 629 & 108                    & 65                   & 97   & 485 \\ \hline
\end{tabular}
\end{adjustbox}
\vspace{3mm}
\end{table}

\begin{figure}[ht] 
\centering
\hspace{-10mm}
\subfigure[{RTM}]
{
\raisebox{-1cm}{\includegraphics[scale=0.25]{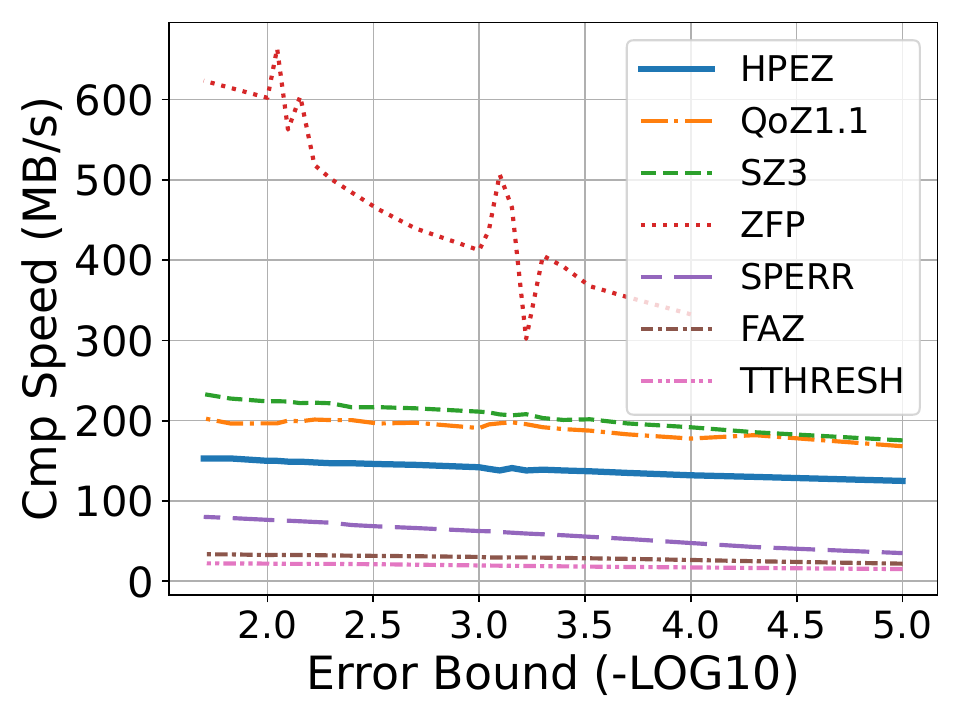}}
}%
\hspace{-2mm}
\subfigure[{CESM-ATM}]
{
\raisebox{-1cm}{\includegraphics[scale=0.25]{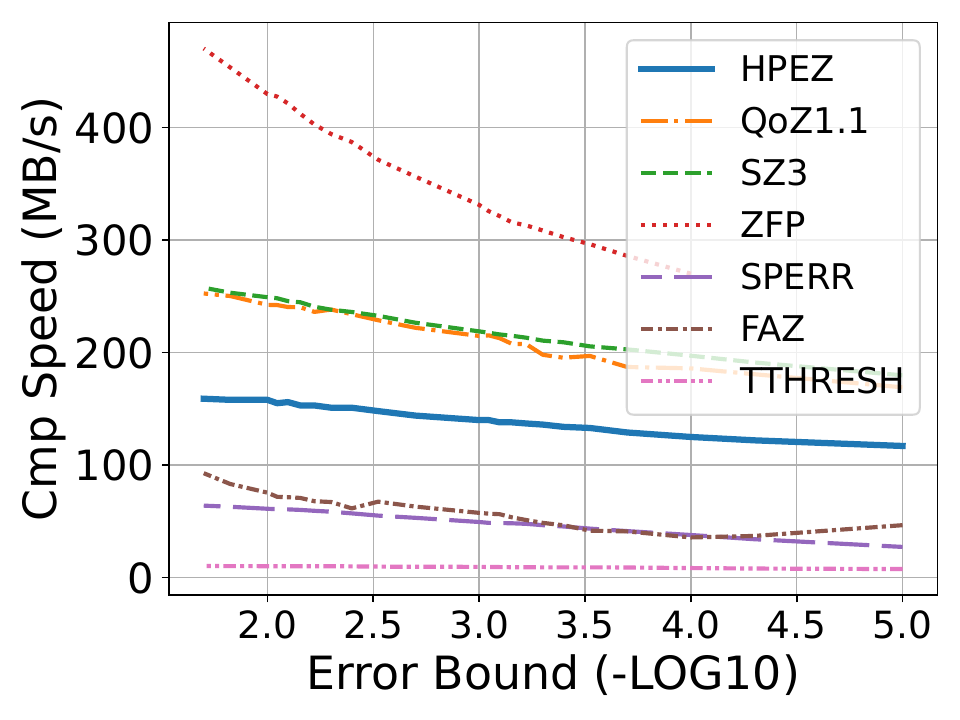}}
}%
\hspace{-2mm}
\subfigure[{JHTDB}]
{
\raisebox{-1cm}{\includegraphics[scale=0.25]{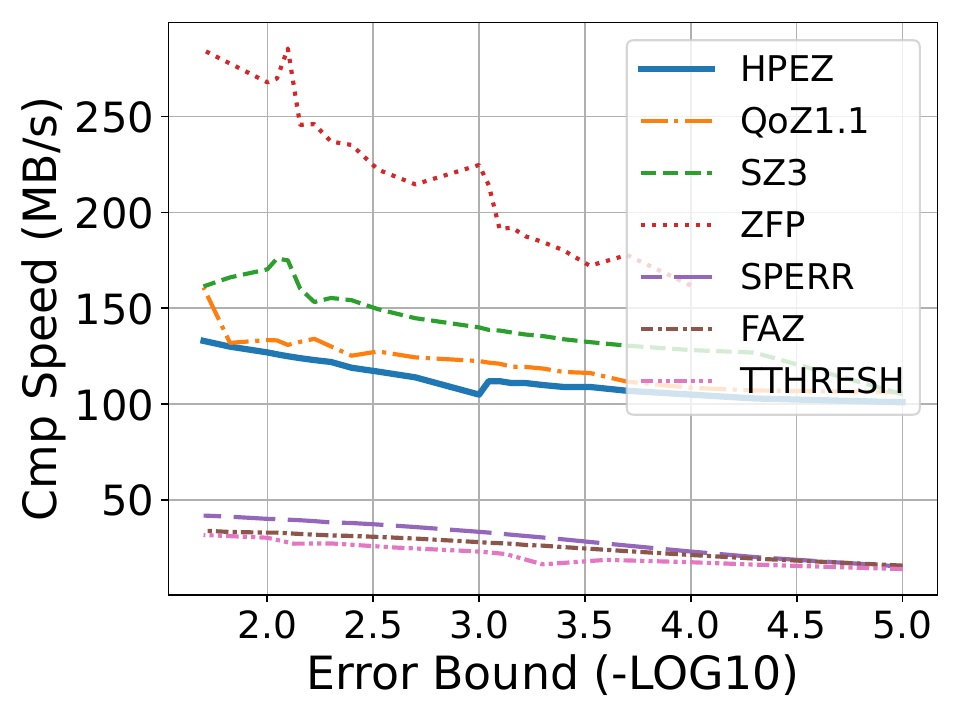}}
}%
\hspace{-10mm}
\vspace{-4mm}

\hspace{-10mm}
\subfigure[{Miranda}]
{
\raisebox{-1cm}{\includegraphics[scale=0.25]{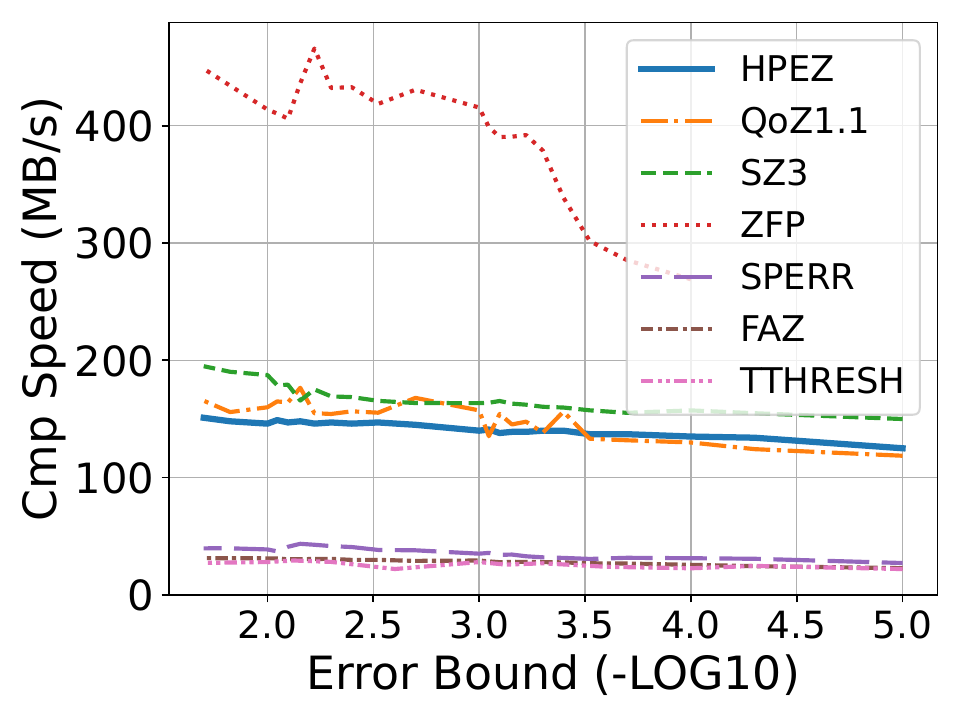}}%
}%
\hspace{-2mm}
\subfigure[{SCALE-LetKF}]
{
\raisebox{-1cm}{\includegraphics[scale=0.25]{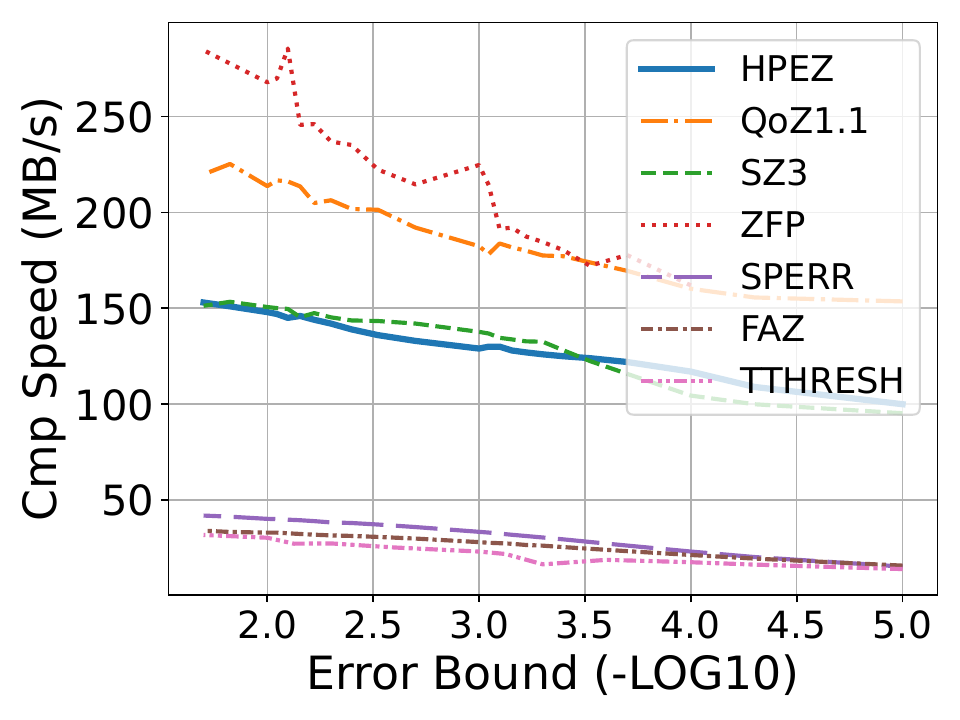}}%
}%
\hspace{-2mm}
\subfigure[{SegSalt}]
{
\raisebox{-1cm}{\includegraphics[scale=0.25]{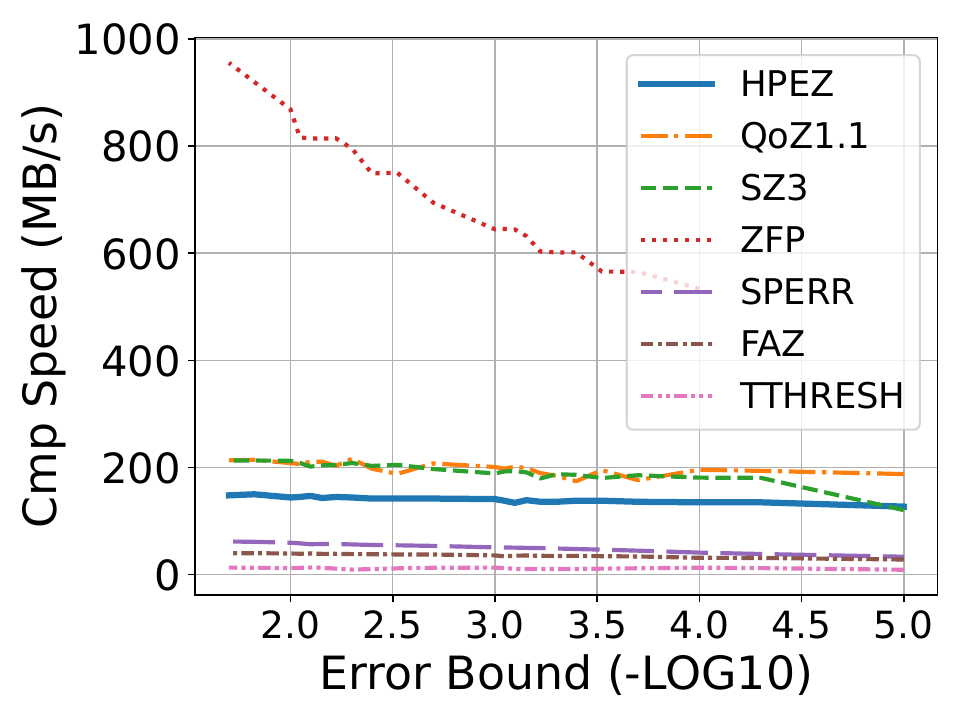}}%
}%
\hspace{-10mm}
\vspace{-2mm}
\caption{Error bound-compression speed plots.}
\label{fig:evaluation-eb-speed}
\end{figure}

\subsubsection{Compression ratios with the same error bounds}
Compressing the datasets with the selected compressors under the same error bounds, we list all the compression ratios in Table \ref{tab:hpcr} and \ref{tab:hrcr}. Table \ref{tab:hpcr} is a comparison among the high-performance compressors (in which the compression ratio optimization targets are set for QoZ, FAZ, and HPEZ). HPEZ achieves the best compression ratio in all cases. On the SegSalt dataset, HPEZ has a $40\%$ $\sim$ $75\%$ compression ratio improvement over the second-best compressor. On the RTM, Miranda, and JHTDB datasets, HPEZ achieves 20\%-45\% compression ratio improvements over the second-best. On the CESM-ATM dataset, under the error bound of 1e-3, HPEZ has a compression ratio of about 2.36$\times$ as high as the second-best (SZ3.1). With these considerable improvements, we can assert that HPEZ is the best choice among high-performance compressors regarding optimizing the error-bound-fixed compression ratio.

We also compare the compression ratios of HPEZ with the ones from the high-ratio compressors in Table \ref{tab:hrcr}. It shows that HPEZ can obtain even higher compression ratios than them in certain cases (e.g. on SCALE-LetKF and JHTDB). Note that the speed of HPEZ is substantially faster than the high-ratio compressors, making it quite competitive over them in speed-concerned use cases.     

\begin{table}[t]
\centering
\footnotesize
  \caption {Compression Ratios of High-Performance Compressors (SZ, ZFP, QoZ and HPEZ)} 
  \vspace{-2mm}
  \label{tab:hpcr} 
  \begin{adjustbox}{width=0.7\columnwidth}
  \begin{tabular}{|c|c|c|c|c|c|c|}
\hline
\multirow{2}{*}{\textbf{Dataset}}  & \multirow{2}{*}{\textbf{$\epsilon$}} & \multirow{2}{*}{\textbf{SZ 3.1}}  & \multirow{2}{*}{\textbf{ZFP 0.5.5}}   & \multirow{2}{*}{\textbf{QoZ 1.1}} & \multirow{2}{*}{\textbf{HPEZ}}  & \multirow{2}{*}{\textbf{Improve (\%)}} \\
                                   &                                   &  &  & & &     \\ \hline
\multirow{3}{*}{\textbf{RTM}}      & 1E-2                              & 1764         & 62.9            & 2156         & \textbf{2701} & 25.3             \\ \cline{2-7} 
                                   & 1E-3                              & 249          & 26.2              & 285          & \textbf{395}  & 38.6             \\ \cline{2-7} 
                                   & 1E-4                              & 55.3         & 14.3            & 58           & \textbf{71.1} & 22.6             \\ \hline
\multirow{3}{*}{\textbf{Miranda}}  & 1E-2                              & 574.6        & 46.6           & 977          & \textbf{1320} & 35.1             \\ \cline{2-7} 
                                   & 1E-3                              & 168          & 25.6           & 181          & \textbf{258}  & 42.5             \\ \cline{2-7} 
                                   & 1E-4                              & 47.3         & 14.5           & 47.7         & \textbf{63.6} & 33.3             \\ \hline
\multirow{3}{*}{\textbf{SegSalt}}  & 1E-2                              & 856          & 59.1           & 1005         & \textbf{1484} & 47.7             \\ \cline{2-7} 
                                   & 1E-3                              & 140.6        & 24.9           & 151          & \textbf{260}  & 72.2             \\ \cline{2-7} 
                                   & 1E-4                              & 38.2         & 14.9           & 35.9         & \textbf{61.7} & 61.5             \\ \hline
\multirow{3}{*}{\textbf{SCALE}}    & 1E-2                              & 167.3        & 14.5           & 160          & \textbf{186}  & 11.2             \\ \cline{2-7} 
                                   & 1E-3                              & 40.4         & 7.8            & 41.5         & \textbf{52.9} & 27.5             \\ \cline{2-7} 
                                   & 1E-4                              & 14.1         & 4.6            & 13.4         & \textbf{15.4} & 9.2              \\ \hline
\multirow{3}{*}{\textbf{JHTDB}}    & 1E-2                              & 528.2        & 22.3           & 647          & \textbf{838}  & 29.5             \\ \cline{2-7} 
                                   & 1E-3                              & 73.2         & 9.8            & 77.8         & \textbf{101}  & 29.8             \\ \cline{2-7} 
                                   & 1E-4                              & 15.8         & 5              & 15.9         & \textbf{20.6} & 29.6             \\ \hline
\multirow{3}{*}{\textbf{CESM-ATM}} & 1E-2                              & 373          & 18.2            & 263          & \textbf{675}  & 81.0             \\ \cline{2-7} 
                                   & 1E-3                              & 64.9         & 9.6              & 59.4         & \textbf{153}  & 135.7            \\ \cline{2-7} 
                                   & 1E-4                              & 22.9         & 5.8            & 21.7         & \textbf{38.9} & 69.9             \\ \hline
\end{tabular}
\end{adjustbox}
\end{table}

\begin{table}[t]
\centering
\footnotesize
  \caption {Compression Ratios of HPEZ and high-ratio compressors (SPERR, FAZ, and TTHRESH)} 
  \vspace{-2mm}
  \label{tab:hrcr} 
  \begin{adjustbox}{width=0.7\columnwidth}
  \begin{tabular}{|c|c|c|c|c|c|}
\hline
\multirow{2}{*}{\textbf{Dataset}}  & \multirow{2}{*}{\textbf{$\epsilon$}} & \multirow{2}{*}{\textbf{SPERR 0.6}} & \multirow{2}{*}{\textbf{FAZ}} & \multirow{2}{*}{\textbf{TTHRESH}} & \multirow{2}{*}{\textbf{HPEZ}} \\
                                   &                                   &                                &                               &                                   &   \\ \hline
\multirow{3}{*}{\textbf{RTM}}      & 1E-2                              & 2187                            & 2695                          & 782                               & \textbf{2701} \\ \cline{2-6} 
                                   & 1E-3                              & 440                             & \textbf{642}                  & 71.4                              & 395           \\ \cline{2-6} 
                                   & 1E-4                              & 84.1                            & \textbf{119}                  & 23.7                              & 71.1          \\ \hline
\multirow{3}{*}{\textbf{Miranda}}  & 1E-2                              & 971.4                           & 996.5                         & 447                               & \textbf{1320} \\ \cline{2-6} 
                                   & 1E-3                              & 243.9                           & \textbf{263.5}                & 142                               & 258           \\ \cline{2-6} 
                                   & 1E-4                              & 74.5                            & \textbf{93.6}                 & 55.1                              & 63.6          \\ \hline
\multirow{3}{*}{\textbf{SegSalt}}  & 1E-2                              & 1219.4                          & \textbf{1639.6}               & 291                               & 1484          \\ \cline{2-6} 
                                   & 1E-3                              & 228.9                           & \textbf{388.9}                & 99.5                              & 260           \\ \cline{2-6} 
                                   & 1E-4                              & 61.3                            & \textbf{117.3}                & 28.8                              & 61.7          \\ \hline
\multirow{3}{*}{\textbf{SCALE}}    & 1E-2                              & 103.5                           & 177.9                         & 80.0                              & \textbf{186}  \\ \cline{2-6} 
                                   & 1E-3                              & 35.5                            & 51.8                          & 18.9                              & \textbf{52.9} \\ \cline{2-6} 
                                   & 1E-4                              & 15                              & \textbf{16.8}                 & 8.4                               & 15.4          \\ \hline
\multirow{3}{*}{\textbf{JHTDB}}    & 1E-2                              & 639.8                           & 726                           & 373                               & \textbf{838}  \\ \cline{2-6} 
                                   & 1E-3                              & 89.3                            & 90.7                          & 65.1                              & \textbf{101}  \\ \cline{2-6} 
                                   & 1E-4                              & 19.9                            & 20.2                          & 17.1                              & \textbf{20.6} \\ \hline
\multirow{3}{*}{\textbf{CESM-ATM}} & 1E-2                              & \textbf{1221}                   & 292                           & 83.5& 675           \\ \cline{2-6} 
                                   & 1E-3                              & 150                             & 77.4                          & 20.4                              & \textbf{153}  \\ \cline{2-6} 
                                   & 1E-4                              & 35                              & 26.3                          & 8.7                               & \textbf{38.9} \\ \hline
\end{tabular}
\end{adjustbox}
\end{table}

\subsubsection{Compression rate-distortion}
\label{sec:rd}
In this section, we mainly present the evaluations of the compression rate-distortion with HPEZ and other high-performance compressors. The high-ratio compressors are capable of achieving excellent compression rate-distortion by spending much more time cost than high-performance compressors, therefore the comparison of rate-distortion would be fair if and only if we exclude the high-ratio compressors, making it within the scope of high-performance compressors to clearly examine how HPEZ has improved the compression quality meanwhile maximally preserving the compression efficiency.

In Figure \ref{fig:evaluation-rate-psnr}, the bit rate-PSNR curves of 4 high-performance compressors on 6 datasets are plotted and displayed (in which the rate-PSNR optimization targets are set for QoZ, FAZ, and HPEZ). Apparently, HPEZ has dominated this evaluation term, achieving the best PSNR under all bit rates on each dataset. This implies that, among the high-performance compressors, HPEZ can always provide the best quality of decompressed data (in terms of PSNR) under the same compression ratio, or can always yield the most compact compressed data for a certain PSNR constraint. On the CESM-ATM dataset, under PSNR of 70, HPEZ reaches around 360\% compression ratio improvement over the second-best QoZ 1.1. On the SegSalt dataset, under PSNR of 80 HPEZ achieves about 100\% compression ratio improvement over the second-best QoZ 1.1. There are approximately 20\% $\sim$ 80\% same-PSNR compression ratio improvements achieved by HPEZ on the other 4 datasets.
\begin{figure}[ht] 
\centering
\hspace{-10mm}
\subfigure[{RTM}]
{
\raisebox{-1cm}{\includegraphics[scale=0.25]{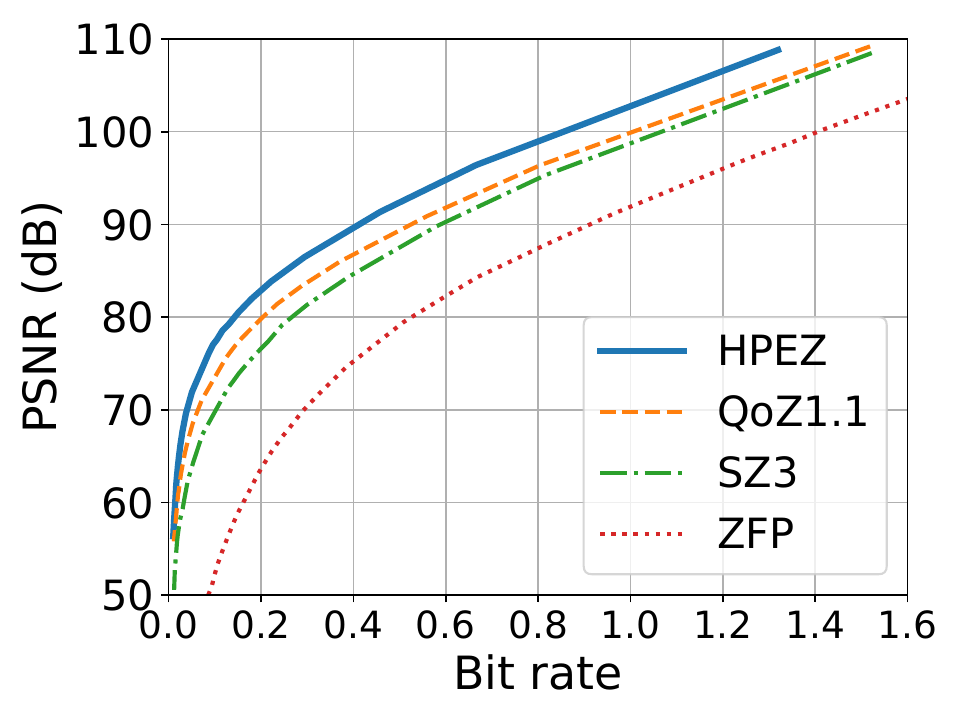}}
}%
\hspace{-2mm}
\subfigure[{CESM-ATM}]
{
\raisebox{-1cm}{\includegraphics[scale=0.25]{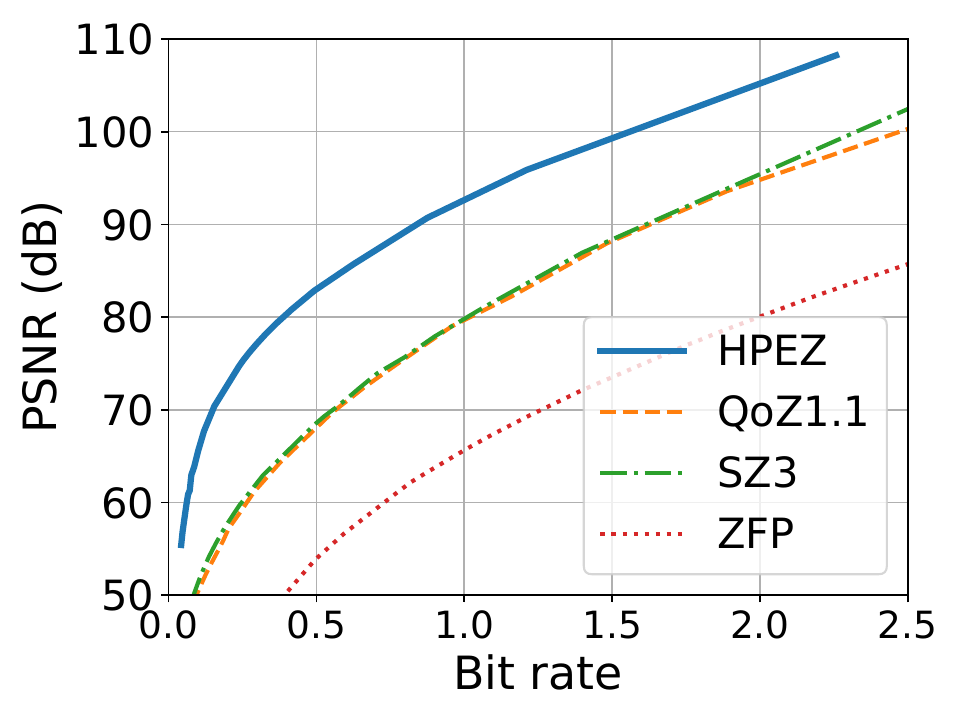}}
}%
\hspace{-2mm}
\subfigure[{JHTDB}]
{
\raisebox{-1cm}{\includegraphics[scale=0.25]{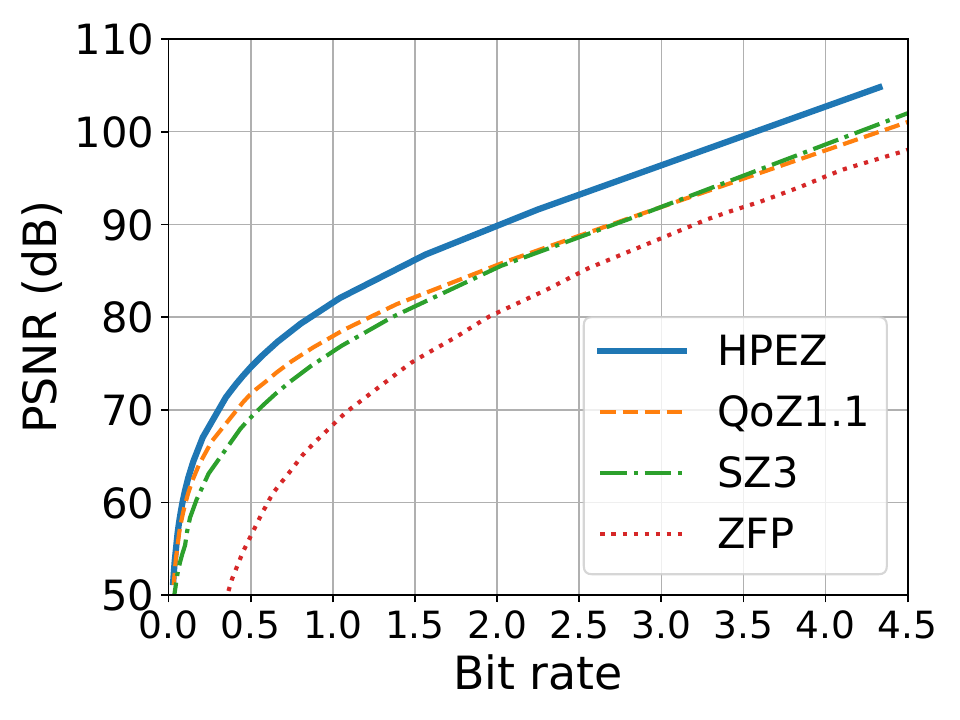}}
}%
\hspace{-10mm}
\vspace{-4mm}

\hspace{-10mm}
\subfigure[{Miranda}]
{
\raisebox{-1cm}{\includegraphics[scale=0.25]{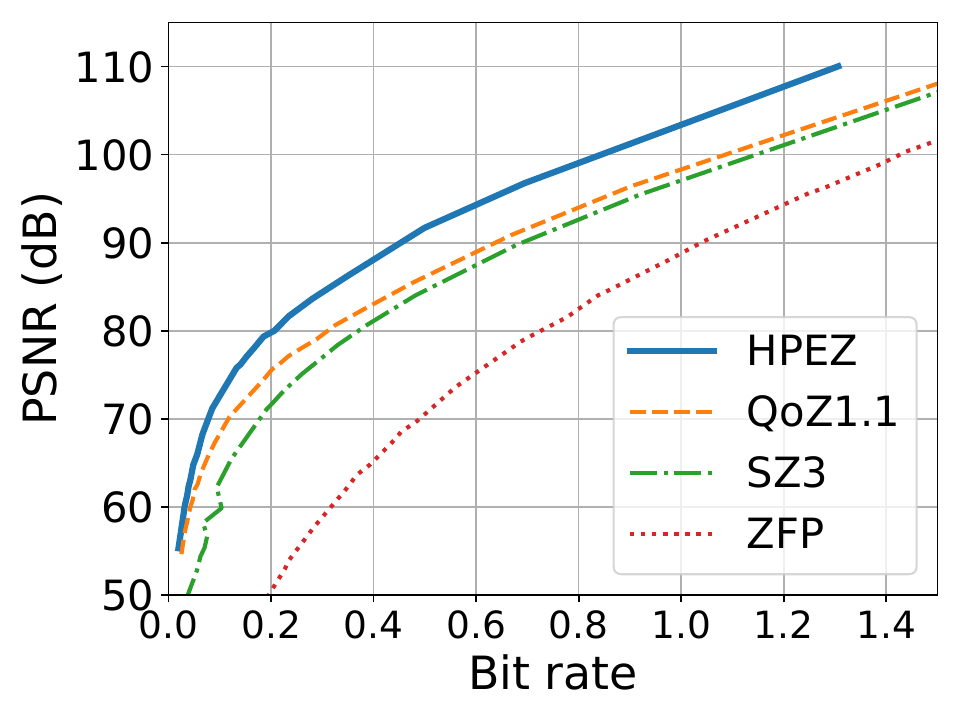}}%
}%
\hspace{-2mm}
\subfigure[{SCALE-LetKF}]
{
\raisebox{-1cm}{\includegraphics[scale=0.25]{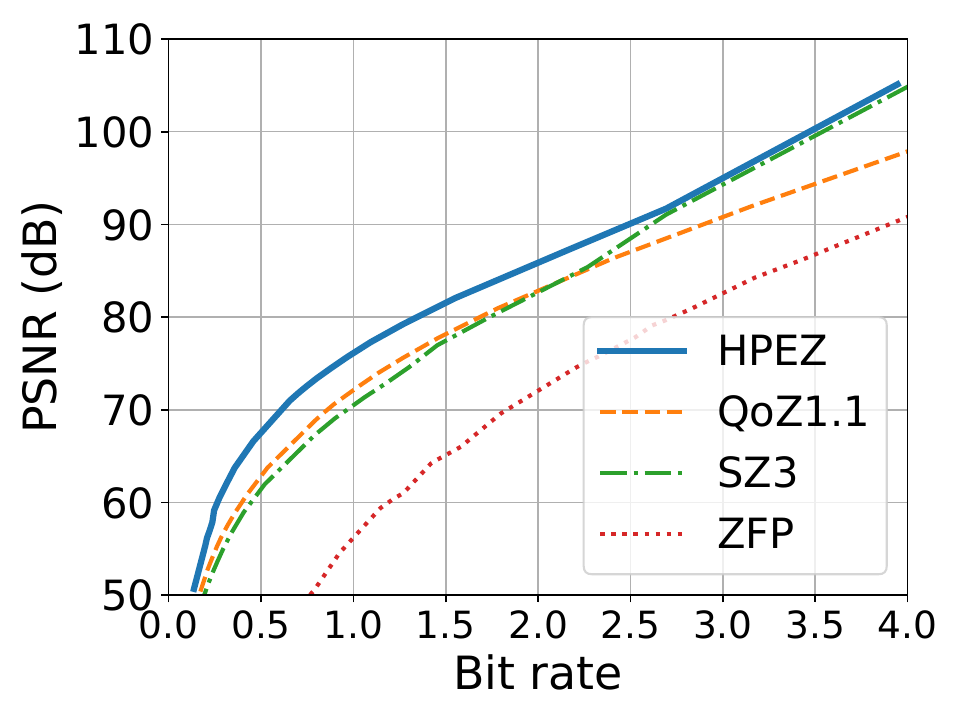}}%
}%
\hspace{-2mm}
\subfigure[{SegSalt}]
{
\raisebox{-1cm}{\includegraphics[scale=0.25]{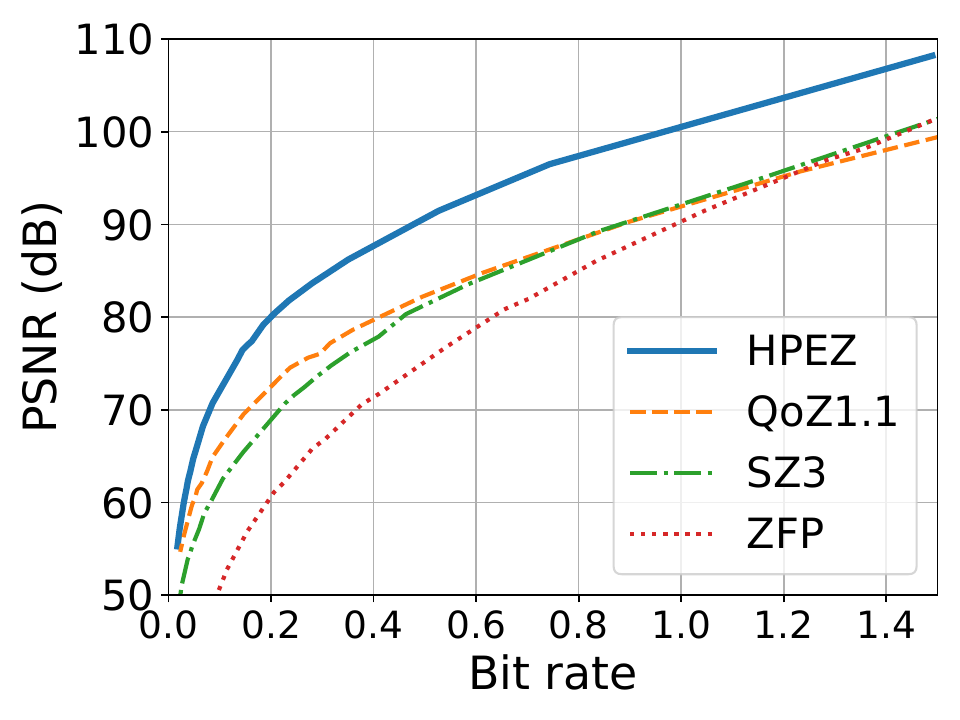}}%
}%
\hspace{-10mm}
\vspace{-2mm}

\caption{Rate-distortion (PSNR) plots of high-performance compressors. }
\label{fig:evaluation-rate-psnr}
\end{figure}

To evaluate the HPEZ compression quality with more quality metrics, we also checked the SSIM of the decompressed results of each high-performance compressor, and those results are illustrated in Figure \ref{fig:evaluation-rate-ssim}. Same as the PSNR, HPEZ undoubtedly presented the best SSIM under the same compressed size over all other evaluated high-performance compressors. Under the same compression bit rate, on multiple datasets including RTM, JHTDB, SCALE-LetKF, and SegSalt, there are 20\% $\sim$ 40\% SSIM improvements from HPEZ over the second-best QoZ 1.1. The SSIM improvements can be even much larger in the case of the CESM-ATM dataset.
\begin{figure}[ht] 
\centering
\hspace{-10mm}
\subfigure[{RTM}]
{
\raisebox{-1cm}{\includegraphics[scale=0.25]{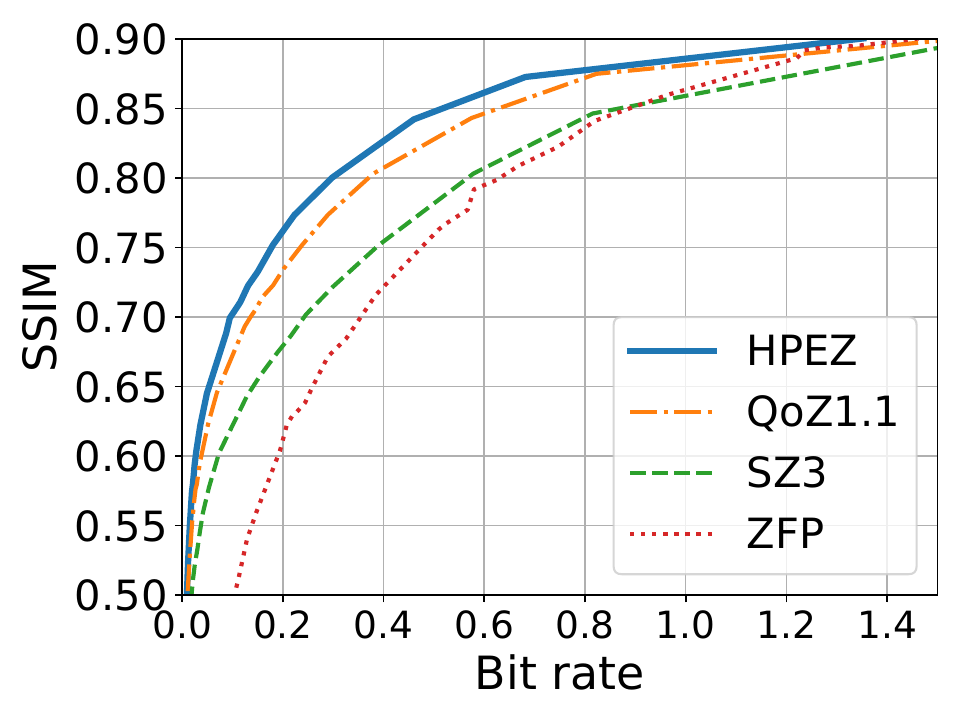}}
}
\hspace{-3mm}
\subfigure[{CESM-ATM}]
{
\raisebox{-1cm}{\includegraphics[scale=0.25]{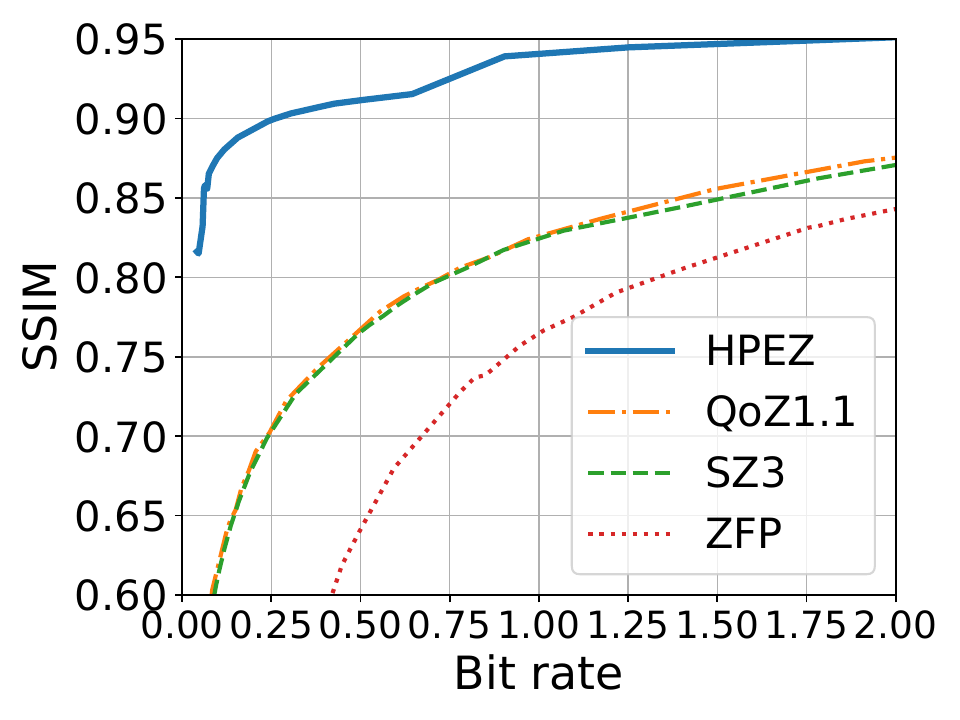}}
}
\hspace{-2mm}
\subfigure[{JHTDB}]
{
\raisebox{-1cm}{\includegraphics[scale=0.25]{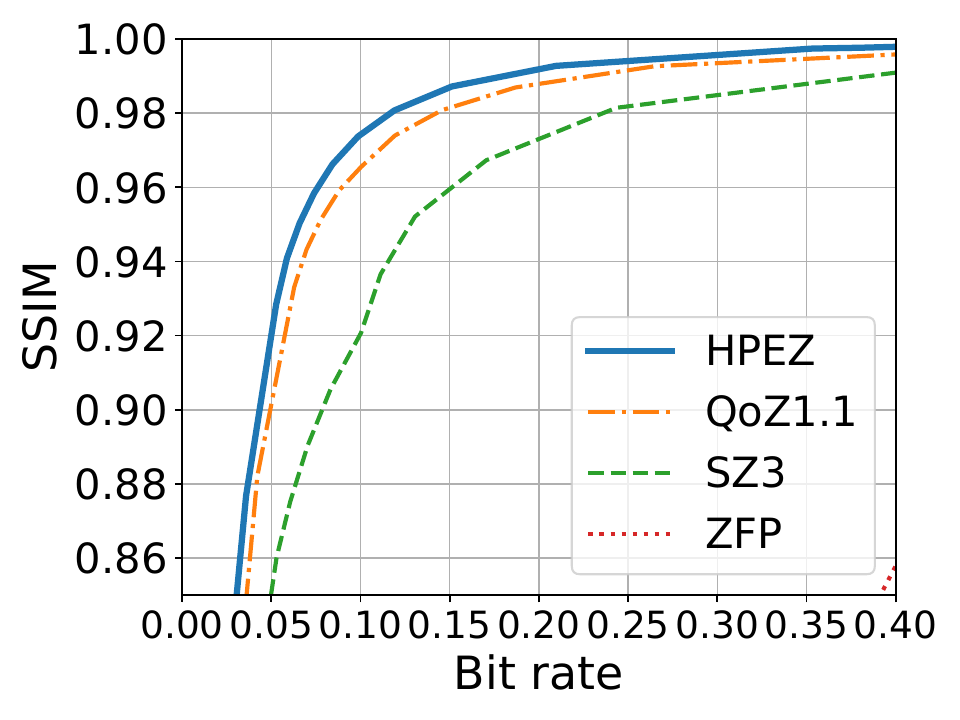}}
}
\hspace{-10mm}
\vspace{-4mm}

\hspace{-10mm}
\subfigure[{Miranda}]
{
\raisebox{-1cm}{\includegraphics[scale=0.25]{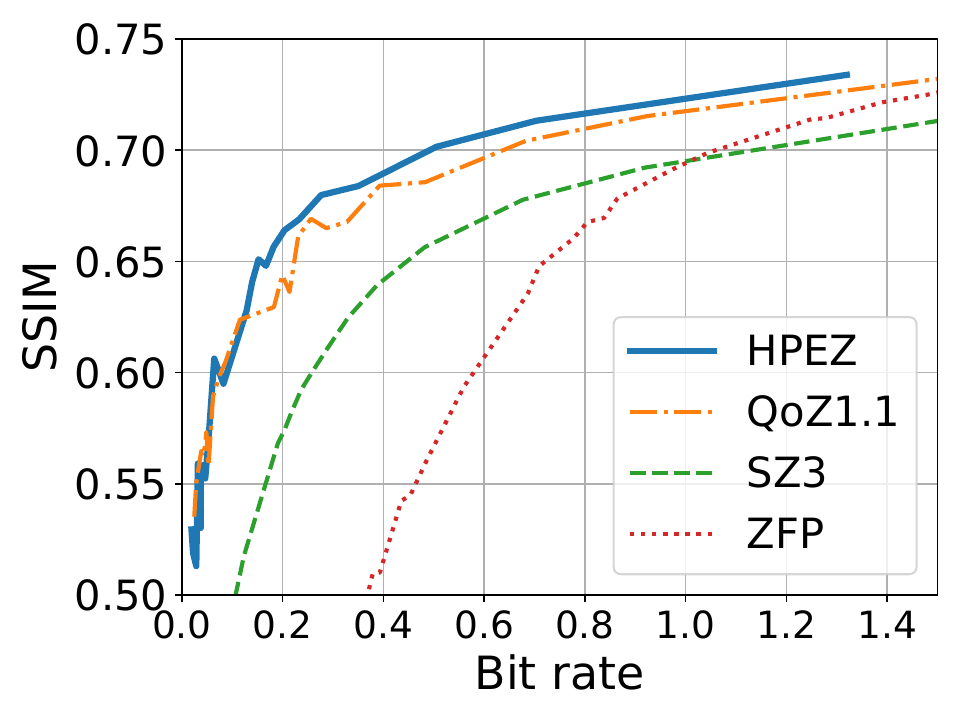}}%
}
\hspace{-2mm}
\subfigure[{SCALE-LetKF}]
{
\raisebox{-1cm}{\includegraphics[scale=0.25]{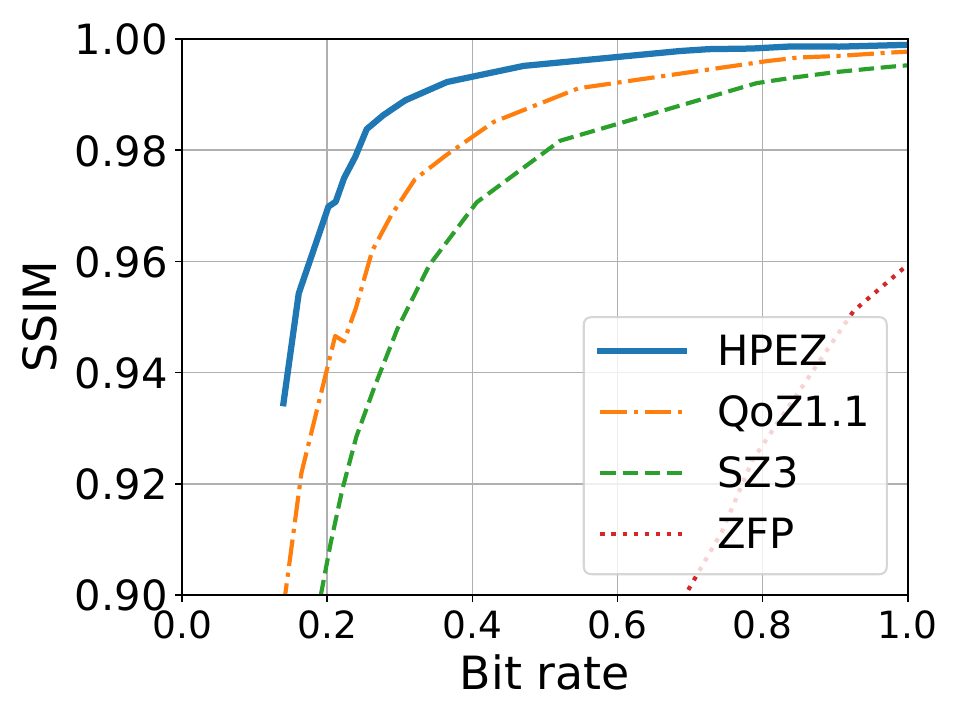}}%
}
\hspace{-2mm}
\subfigure[{SegSalt}]
{
\raisebox{-1cm}{\includegraphics[scale=0.25]{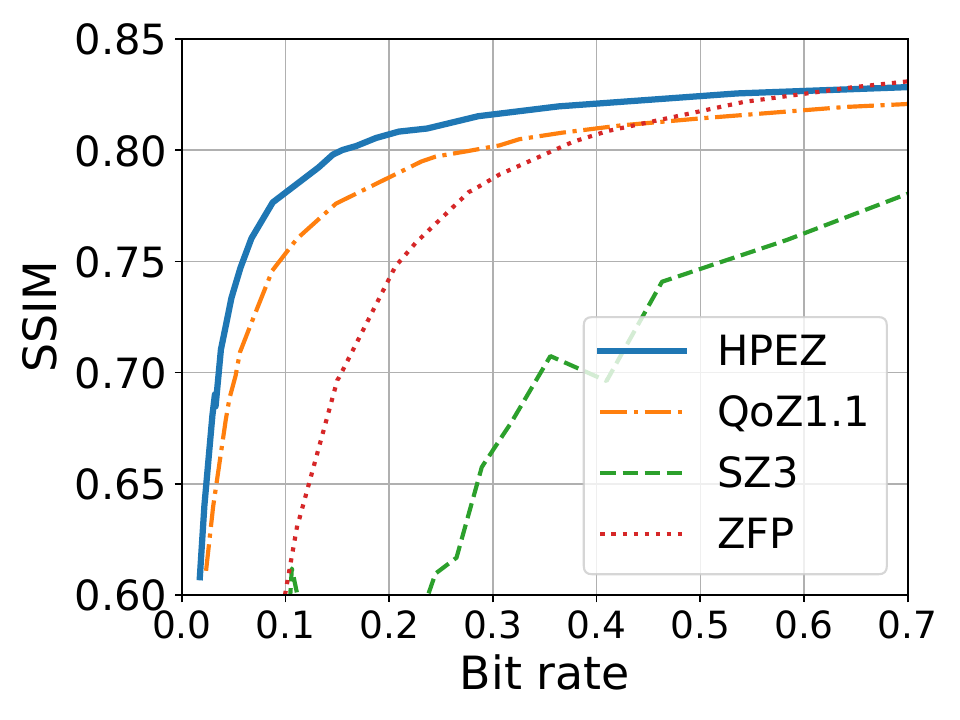}}%
}
\hspace{-10mm}
\vspace{-2mm}

\caption{SSIM of high-performance compressors.}
\label{fig:evaluation-rate-ssim}
\end{figure}

In our analysis, the outstanding compression quality of HPEZ as a high-performance compressor is attributed to 3 aspects: First, the advanced interpolation techniques described in Section \ref{sec:interp} have significantly raised the interpolation-based prediction accuracy for smooth datasets such as RTM, Miranda, SegSalt, and JHTDB. Next, avoiding interpolations along non-smooth directions, the compression for datasets with non-smooth dimensions (e.g. SCALE-LetKF and CESM-ATM) have been obviously enhanced by the dynamic dimension freezing technique (Section \ref{sec:freeze}). Third, the block-wise interpolation tuning (Section \ref{sec:blockwise}) fine-tunes the interpolation on each data block, further optimizing the overall compression. In Section \ref{sec:abla}, we will feature the contribution of each HPEZ design component with more experimental results and in-depth analysis. 

Lastly, we would like to claim that, in several cases, the compression quality (i.e. rate-distortion) of HPEZ can be at least comparable with the high-ratio compressors. In the comparisons between HPEZ and high-ratio compressors displayed in Figure \ref{fig:evaluation-rate-psnr-hr}, although on the Miranda dataset (Figure \ref{fig:evaluation-rate-psnr-hr} (b)) HPEZ has quality gaps to the SPERR and FAZ, Figure \ref{fig:evaluation-rate-psnr-hr} (a), (c) and (d) indicate that HPEZ may achieve close or even similar rate-distortion to the high-ratio compressors, with a compression speed far higher than them.
\begin{figure}[ht] 
\centering
\hspace{-10mm}
\subfigure[{JHTDB}]
{
\raisebox{-1cm}{\includegraphics[scale=0.25]{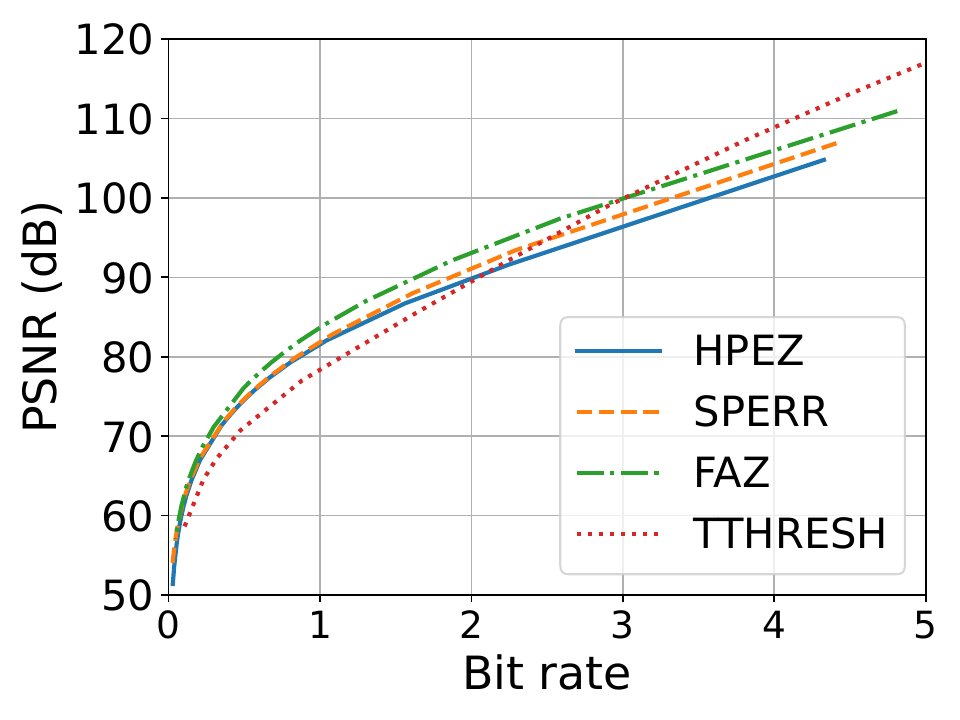}}
}
\hspace{-4mm}
\subfigure[{Miranda}]
{
\raisebox{-1cm}{\includegraphics[scale=0.25]{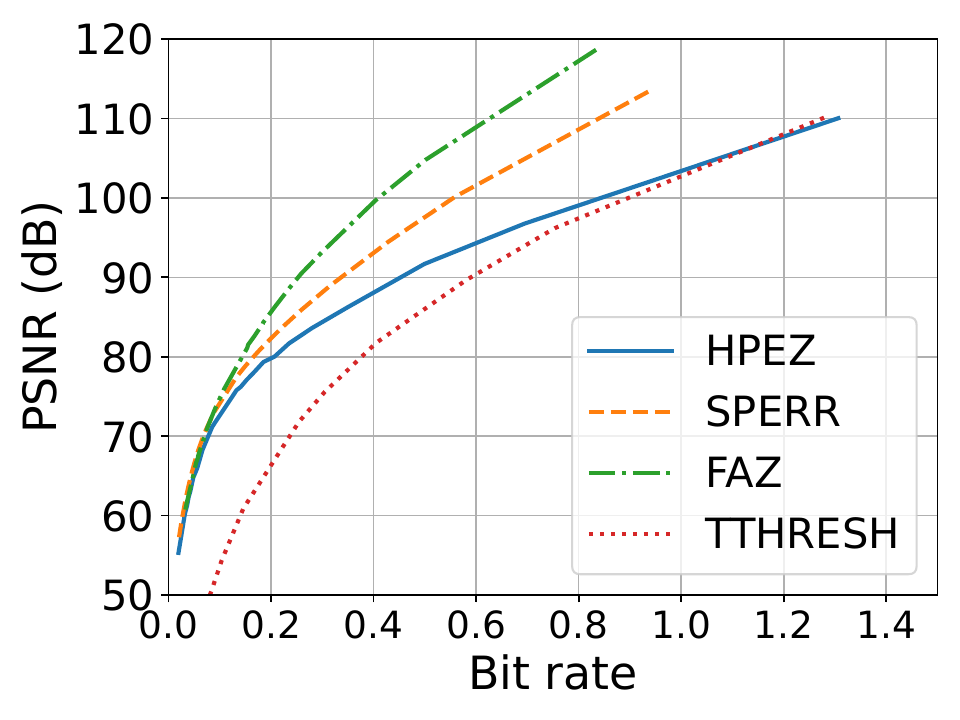}}%
}
\hspace{-10mm}
\vspace{-4mm}

\hspace{-10mm}
\subfigure[{SCALE-LetKF}]
{
\raisebox{-1cm}{\includegraphics[scale=0.25]{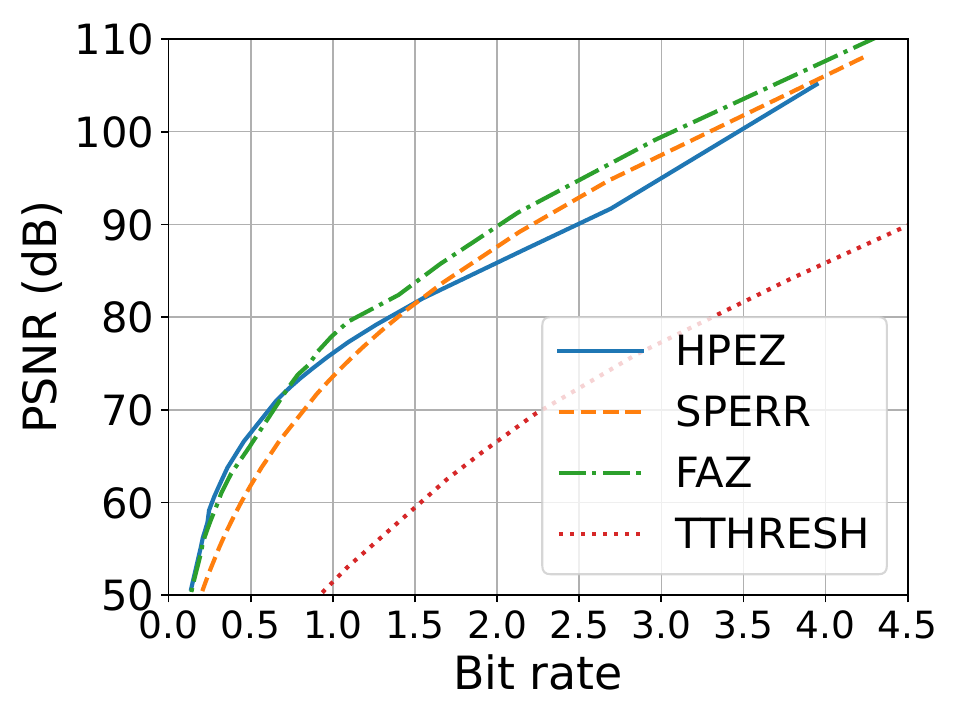}}%
}
\hspace{-4mm}
\subfigure[{CESM-ATM}]
{
\raisebox{-1cm}{\includegraphics[scale=0.25]{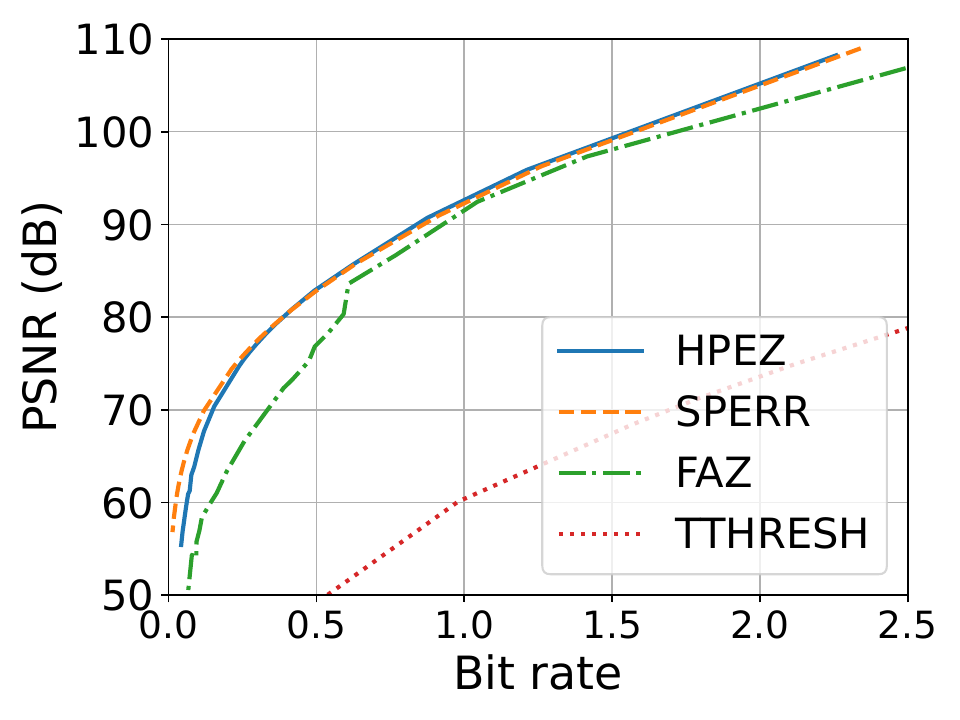}}%
}
\hspace{-10mm}
\vspace{-2mm}

\caption{Rate-PSNR of HPEZ and high-ratio compressors (HPEZ's speed is 2x-17x of high-ratio compressors).}
\label{fig:evaluation-rate-psnr-hr}
\end{figure}

\vspace{-2mm}
\subsubsection{Parallel data transfer test on the distributed database}

In Section \ref{sec:rd}, we have analyzed how HPEZ over-performs other high-performance compressors in terms of compression quality. Furthermore, we will examine whether HPEZ can over-perform existing state-of-the-art error-bounded lossy compressors including high-ratio compressors in real-world use cases in which the compression and decompression time need to be taken into account. To this end, we have designed a real-world scale parallel data transfer experiment on the distributed scientific database. In this experiment, a distributed scientific database is established on multiple machines, and to accomplish the target of fast data transfer and access between the super-computers, instead of costing unacceptable time transferring the original exascale data, a lossy compressor compresses and decompresses the data in parallel on the source and destination machine, and only the compressed data with a highly-reduced size are transferred between the machines. The total time cost of this task is the accumulation of the local data I/O time, compression time, decompression time, and transfer time of the compressed data. 

\begin{table}[t]
\centering
\footnotesize
  \caption {Compression-based parallel data transfer throughput time (in seconds, 2048 cores, under PSNR=80). Inter-machine speed is the transfer speed of compressed data between 2 machines.} 
  \vspace{-2mm}
  \label{tab:eva-para4} 
\begin{adjustbox}{width=0.99\columnwidth}
\begin{tabular}{|c|c|c|c|c|c|c|c|c|c|c|}
\hline
\multirow{2}{*}{Dataset} & \multicolumn{1}{c|}{\multirow{2}{*}{Direction}} & Inter-machine   & \multirow{2}{*}{SZ3} & \multirow{2}{*}{ZFP} & \multirow{2}{*}{QoZ 1.1} & \multirow{2}{*}{SPERR 0.6} & \multirow{2}{*}{FAZ} & \multirow{2}{*}{TTHRESH} & \multirow{2}{*}{HPEZ} & Improve \\
                         & \multicolumn{1}{c|}{}                           & Speed (GB/s)    &                      &                      &                          &                        &                      &                          &                          & (\%)    \\ \hline
CESM-ATM                 & Anvil to Bebop                                  & 0.79 $\sim$0.91 & 1934                 & 3221                 & 1812                     & 1560                   & 1586                 & 7752                     & \textbf{1005}            & 35.6    \\ \cline{2-11} 
(41TB)                   & Bebop to Anvil                                  & 0.95 $\sim$1.19 & 1614                 & 2695                 & 1553                     & 1522                   & 1544                 & 8560                      & \textbf{916}             & 39.8    \\ \hline
RTM                      & Anvil to Bebop                                  & 0.58 $\sim$1.19 & 198                  & 362                  & \textbf{173}             & 277                    & 494                  & 527                      & 181                      & -4.8    \\ \cline{2-11} 
(14TB)                   & Bebop to Anvil                                  & 0.47 $\sim$1.04 & 189                  & 524                  & \textbf{166}             & 296                    & 474                  & 560                      & 182                      & -9.5    \\ \hline
Miranda                  & Anvil to Bebop                                  & 0.46 $\sim$1.04 & 49                   & 84                   & 44                       & 72                     & 87                   & 121                      & \textbf{39}              & 11.3    \\ \cline{2-11} 
(2TB)                    & Bebop to Anvil                                  & 0.54 $\sim$0.82 & 46                   & 117                  & 49                       & 71                     & 86                   & 120                     & \textbf{43}              & 6.5     \\ \hline
SCALE-LetKF              & Anvil to Bebop                                  & 0.88 $\sim$0.94 & 873                  & 1354                 & 820                      & 1037                   & 782                  & 2354                      & \textbf{728}             & 7.0     \\ \cline{2-11} 
(13TB)                   & Bebop to Anvil                                  & 1.05 $\sim$1.15 & 745                  & 1181                 & 707                      & 1007                   & 670                  & 2002                     & \textbf{624}             & 6.8     \\ \hline
JHTDB                    & Anvil to Bebop                                  & 0.83 $\sim$1.15 & 567                  & 826                  & 527                      & 645                    & 583                  & 835                      & \textbf{417}             & 20.9    \\ \cline{2-11} 
(10TB)                   & Bebop to Anvil                                  & 0.97 $\sim$1.18 & 486                  & 707                  & 473                      & 648                    & 574                  & 883                      & \textbf{366}             & 22.7    \\ \hline
SegSalt                  & Anvil to Bebop                                  & 0.63 $\sim$1.18 & 163                  & 289                  & 174                      & 221                    & 251                  & 393                      & \textbf{137}             & 15.9    \\ \cline{2-11} 
(8TB)                    & Bebop to Anvil                                  & 0.76 $\sim$1.06 & 167                  & 241                  & 153                      & 213                    & 265                  & 300                      & \textbf{132}             & 14.0    \\ \hline
\end{tabular}
\end{adjustbox}
\end{table}
To convincingly prove the effectiveness of HPEZ for the parallel data transfer task, we conduct the corresponding experiments under a certain configuration. For a parallel test with $p$ cores, we augment the datasets by $p$ times then let each core compress and decompress the data in the original size. Using 2048 cores, we leveraged the 7 compressors to compress and transfer the datasets bidirectionally between the Anvil and Bebop supercomputer, constraining the decompressed data following the same distortion (PSNR = 80). The inter-machine data transfer is supported by the Globus Transfer Service \cite{globus1,globus2,globus3}, which is an efficient and widely adopted data transfer service in scientific research and education fields. Table \ref{tab:eva-para4} presents data transfer speed and the time cost with each compressor for each dataset. On most of the datasets tested (except for the RTM), HPEZ improves the optimal overall transfer time by 5\% $\sim$ 40\%, and in the worst case (on the RTM dataset), it is just slightly worse than QoZ 1.1. Therefore, the optimized balance of compression quality and efficiency of HPEZ does contribute to its utility in real-world large-scale parallel data transfer tasks.

Due to the computing resource limitation for executing the multi-core large-scale data transfer tests and repeating them with different datasets, compressors, and configurations, we have also designed a model for approximating the actual time costs in those tasks. For a specific core number $p$ and a data transfer speed $s$, we use the sequential compression/decompression speed of the compression/decompression with the same per-core data to estimate the parallel compression/decompression time cost, and the approximated data transfer time is just the compressed data size divides the transfer speed $s$. With those approximation methods, for each dataset, we approximate the time costs under a variety of compression error bounds, then plot and present the time cost-PSNR curves in Figure \ref{fig:evaluation-para}. The compressor speeds are acquired on the Anvil machine introduced in Section \ref{sec:env}, the core numbers are 2048, and the data transfer speed is set to 1GB/s (according to the experimental results in Table \ref{tab:eva-para4}). From the plots, we can claim that, for this task, HPEZ has the potential to keep an advantage over the other existing compressors. On Miranda, CESM-ATM, and JHTDB datasets, with the approximations, HPEZ exhibits the minimized time cost for each fixed PSNR. On RTM, SCALE-LetKF, and SegSalt datasets, HPEZ may still always be the top-performing compressor and can have the optimized time cost in wide ranges of PSNR.

\begin{figure}[ht] 
\centering
\hspace{-10mm}
\subfigure[{RTM}]
{
\raisebox{-1cm}{\includegraphics[scale=0.25]{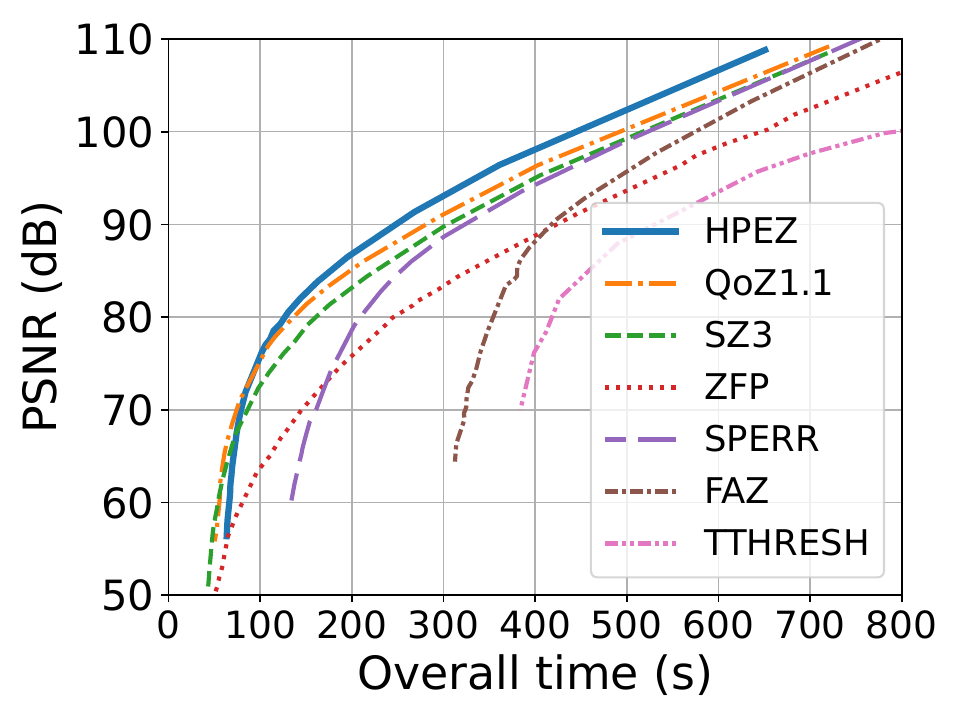}}
}
\hspace{-2mm}
\subfigure[{CESM-ATM}]
{
\raisebox{-1cm}{\includegraphics[scale=0.25]{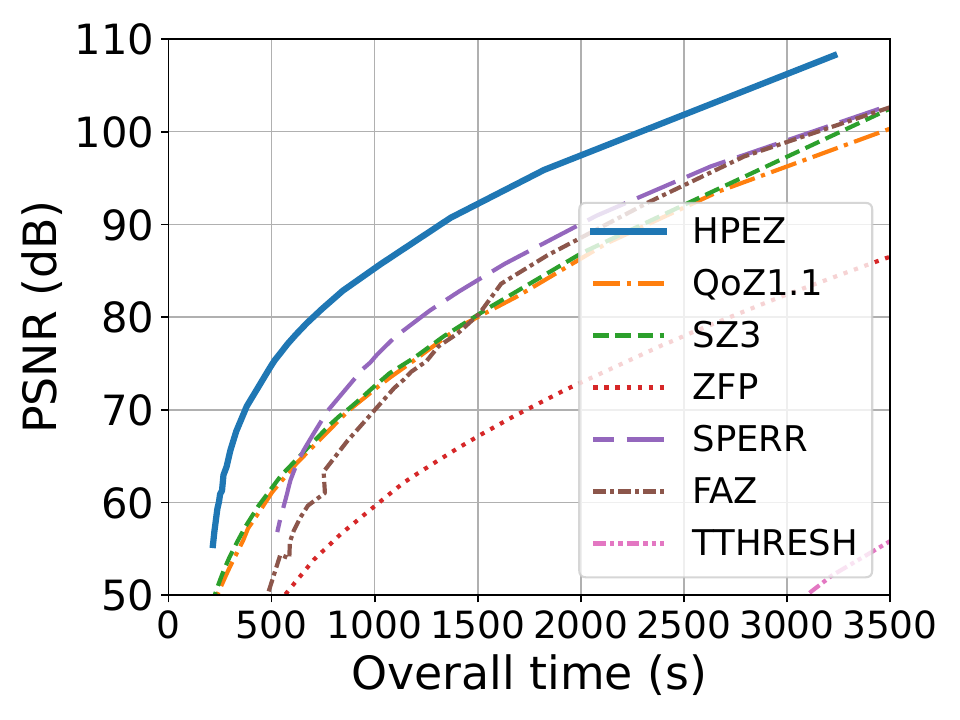}}
}
\hspace{-2mm}
\subfigure[{JHTDB}]
{
\raisebox{-1cm}{\includegraphics[scale=0.25]{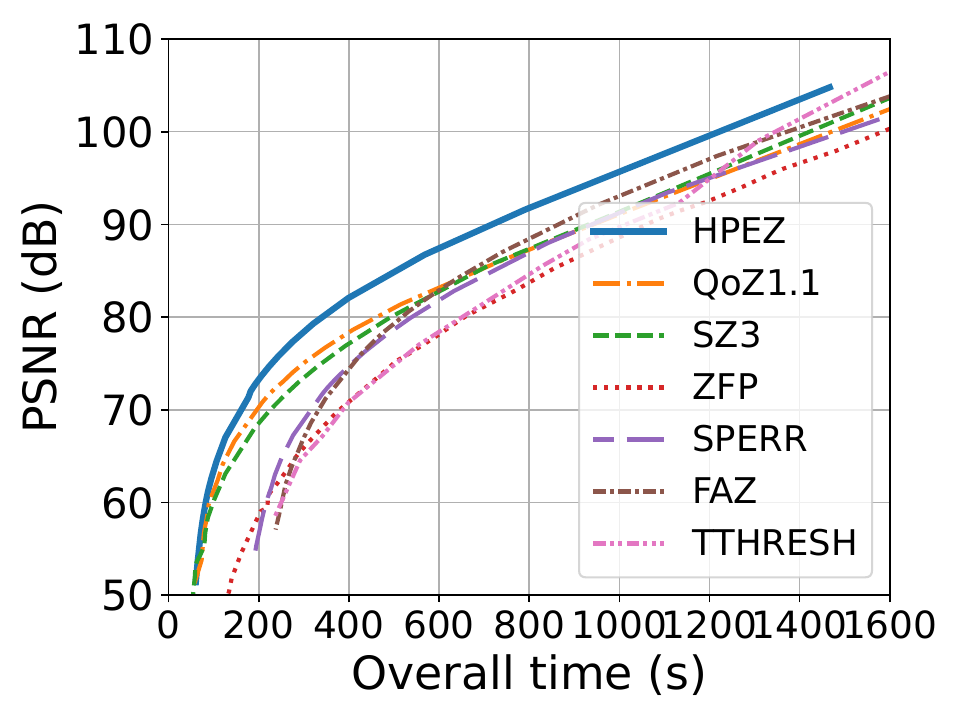}}
}
\hspace{-10mm}
\vspace{-4mm}

\hspace{-10mm}
\subfigure[{Miranda}]
{
\raisebox{-1cm}{\includegraphics[scale=0.25]{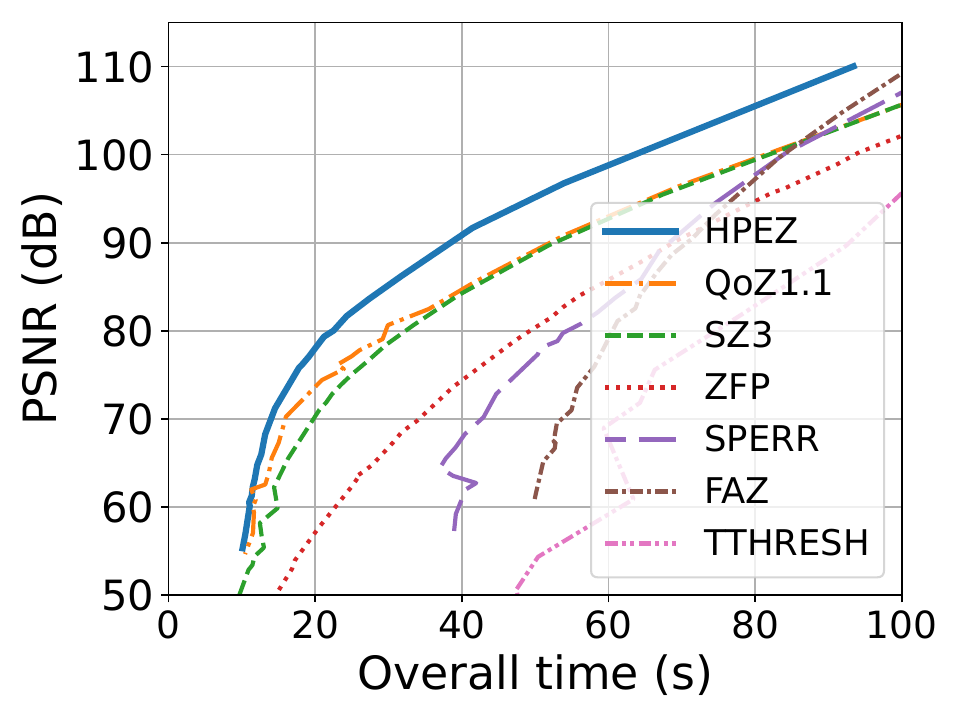}}%
}
\hspace{-2mm}
\subfigure[{SCALE-LetKF}]
{
\raisebox{-1cm}{\includegraphics[scale=0.25]{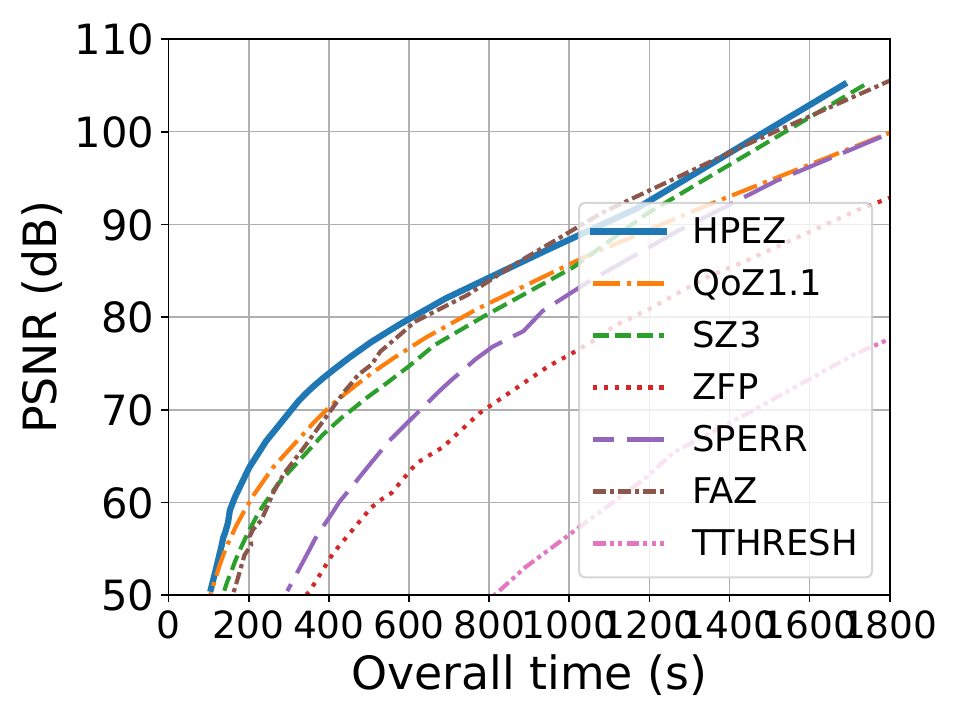}}%
}
\hspace{-2mm}
\subfigure[{SegSalt}]
{
\raisebox{-1cm}{\includegraphics[scale=0.25]{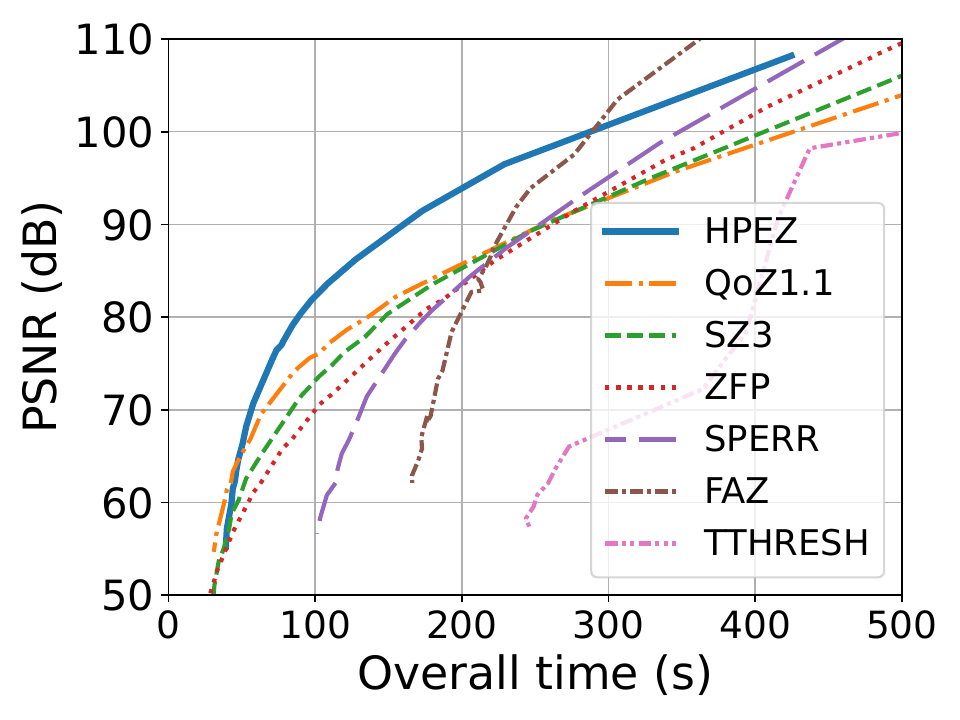}}%
}
\hspace{-10mm}
\vspace{-2mm}

\caption{Parallel data transfer time approximation and decompression PSNR (simulation on the Anvil supercomputer, $p=2048$, $s=1GB/s$).}
\label{fig:evaluation-para}
\end{figure}

\subsubsection{Case study: decompression visualizations}
As an example of the effectiveness of the HPEZ compression, in this section, we propose a case study of the compression tasks, visualizing the decompression outputs from various high-performance compressors. The example data input is the QS field (getting logarithmized in preprocessing) from the SCALE-LetKF dataset, and we compress it with HPEZ and 2 high-performance compressors: QoZ and ZFP (we omit SZ3 in this test because QoZ and SZ3 have close speeds and QoZ has better compression quality than SZ3) under similar compression ratios. The visualizations of 2-D slices from the original data and decompressed data are presented in Figure \ref{fig:evaluation-vis}. In this case, among the decompression results with very close compression ratios, the decompression result of HPEZ (Figure \ref{fig:evaluation-vis} (b)) achieves the lowest data distortion with the highest PSNR (56.8). Moreover, regarding the magnified regions in Figure \ref{fig:evaluation-vis}, compared to the decompression results of QoZ (Figure \ref{fig:evaluation-vis} (c), PSNR=52.7), HPEZ has better preserved the local data patterns in the original input (Figure \ref{fig:evaluation-vis} (a)). This case is an example to show the strong capability of HPEZ in providing high-quality compression results with high compression speed.

\begin{figure}[ht] 
\hspace{-9mm}
\subfigure[{SCALE-LetKF (QS)}]
{
\raisebox{-1cm}{\includegraphics[scale=0.35]{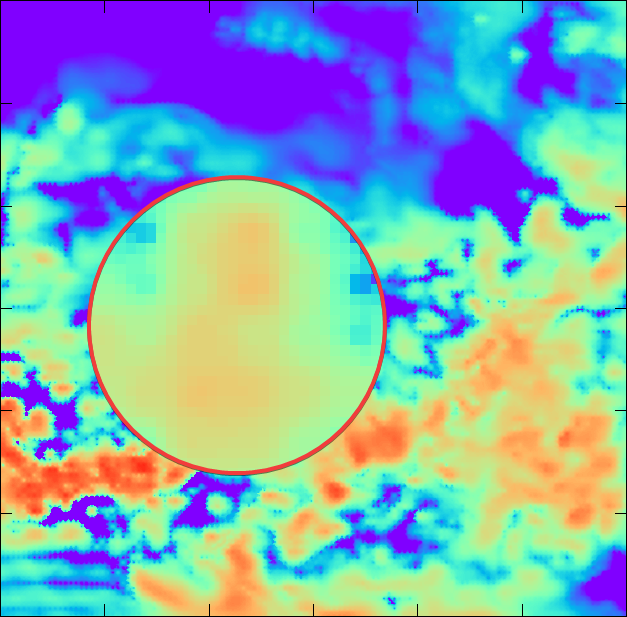}}
}
\hspace{-2mm}
\subfigure[{HPEZ (CR=127,PSNR=56.8)}]
{
\raisebox{-1cm}{\includegraphics[scale=0.35]{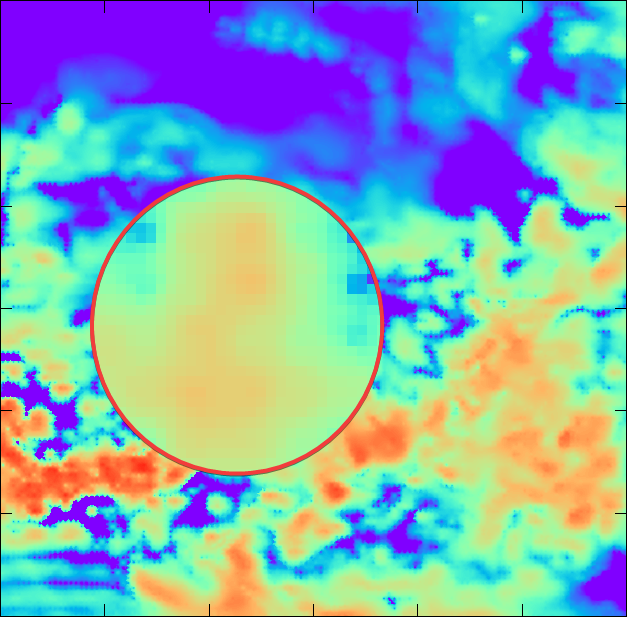}}
}
\hspace{-10mm}
\vspace{-2mm}

\hspace{-10mm}
\subfigure[{QoZ (CR=126,PSNR=48.4)}]
{
\raisebox{-1cm}{\includegraphics[scale=0.35]{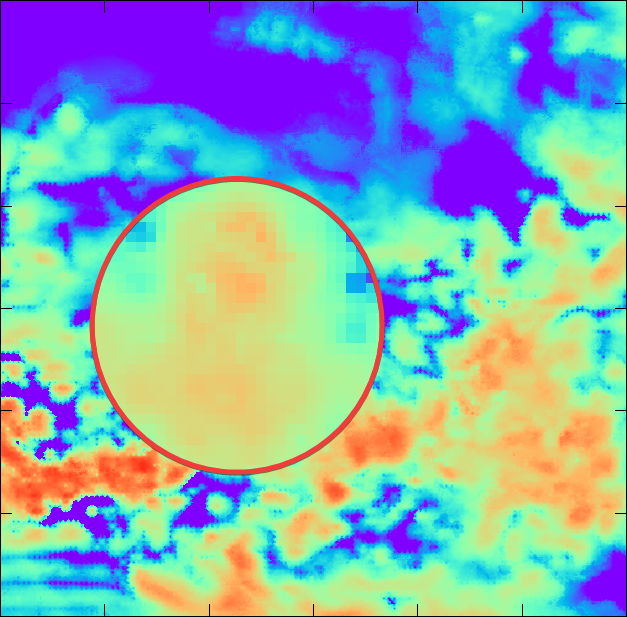}}
}
\hspace{-2mm}
\subfigure[{ZFP (CR=118,PSNR=21.0)}]
{
\raisebox{-1cm}{\includegraphics[scale=0.35]{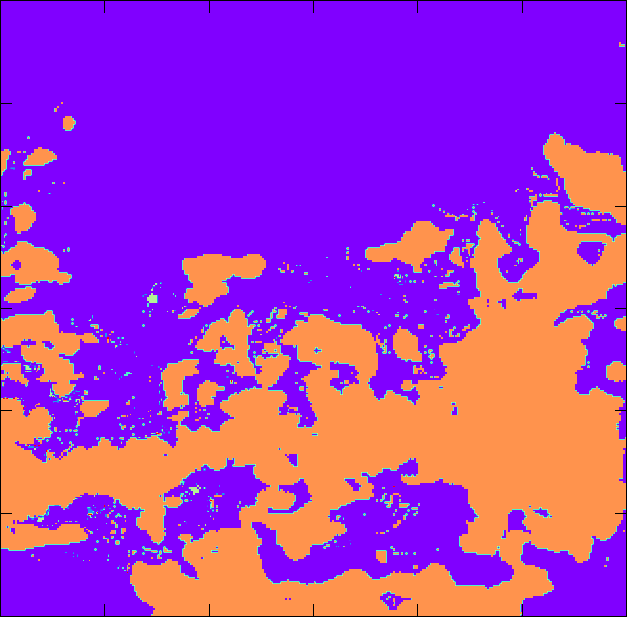}}%
}
\hspace{-10mm}
\vspace{-3mm}

\caption{Visualization of SCALE-QS field (logarithmized) and the decompressed data. }
\label{fig:evaluation-vis}
\end{figure}

\vspace{-2mm}
\subsubsection{Compression of HPEZ on integer datasets}
In this section, we propose the compression rate-distortion of HPEZ on the 2 integer datasets described in Section \ref{sec:setup}. Those datasets are scientific images and movies, therefore the experimental results with them can also reflect the potential of HPEZ to be leveraged on more integer-based datasets such as natural images and videos. Figure \ref{fig:evaluation-rate-psnr-int} contains the rate-PSNR curves from HPEZ and other integer-supportive high-performance compressors (SZ3 and QoZ). Apparently, HPEZ has comparatively excellent rate-distortion on the integer datasets as well as on the floating point datasets, presenting the optimized or near-optimized PSNR under the same bit rate.

\begin{figure}[ht] 
\centering
\hspace{-10mm}
\subfigure[{APS (2D Image)}]
{
\raisebox{-1cm}{\includegraphics[scale=0.25]{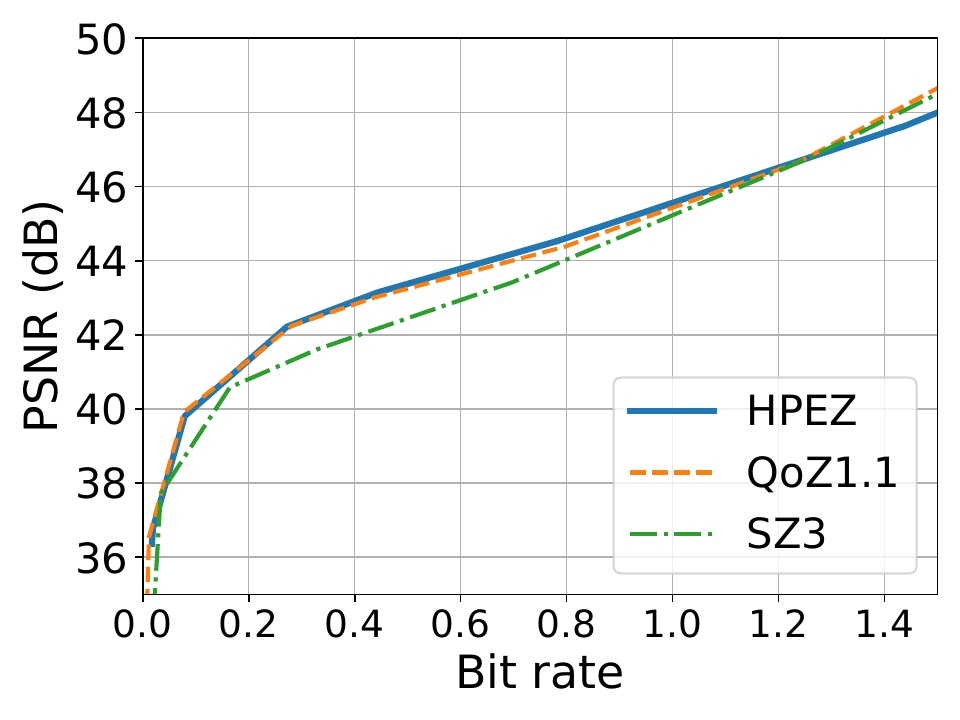}}
}
\hspace{-4mm}
\subfigure[{NSTX-GPI (3D Movie)}]
{
\raisebox{-1cm}{\includegraphics[scale=0.25]{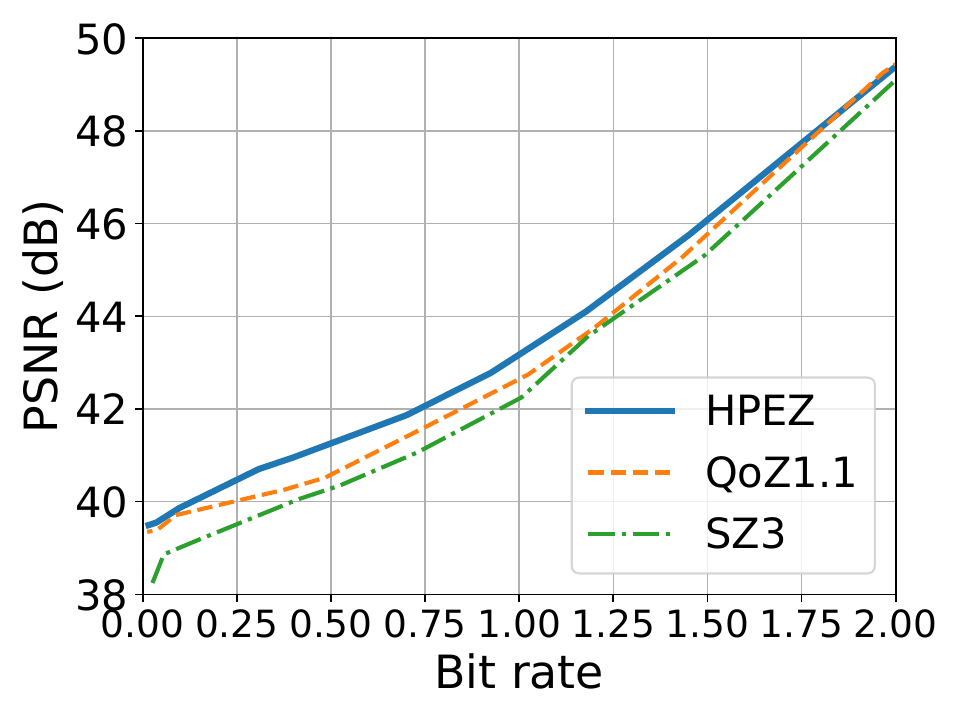}}%
}
\hspace{-10mm}
\vspace{-2mm}
\caption{Rate-PSNR on integer datasets.}
\label{fig:evaluation-rate-psnr-int}
\end{figure}
\subsubsection{Ablation study}
\label{sec:abla}
To better understand how HPEZ can generate high-quality compression outputs with comparatively fast speeds, we decompose the design of HPEZ, aggregating the design components to QoZ 1.1 one by one for determining and quantifying the compression improvement brought by each component.

\begin{figure}[ht] 
\centering
\hspace{-10mm}
\subfigure[{RTM}]
{
\raisebox{-1cm}{\includegraphics[scale=0.25]{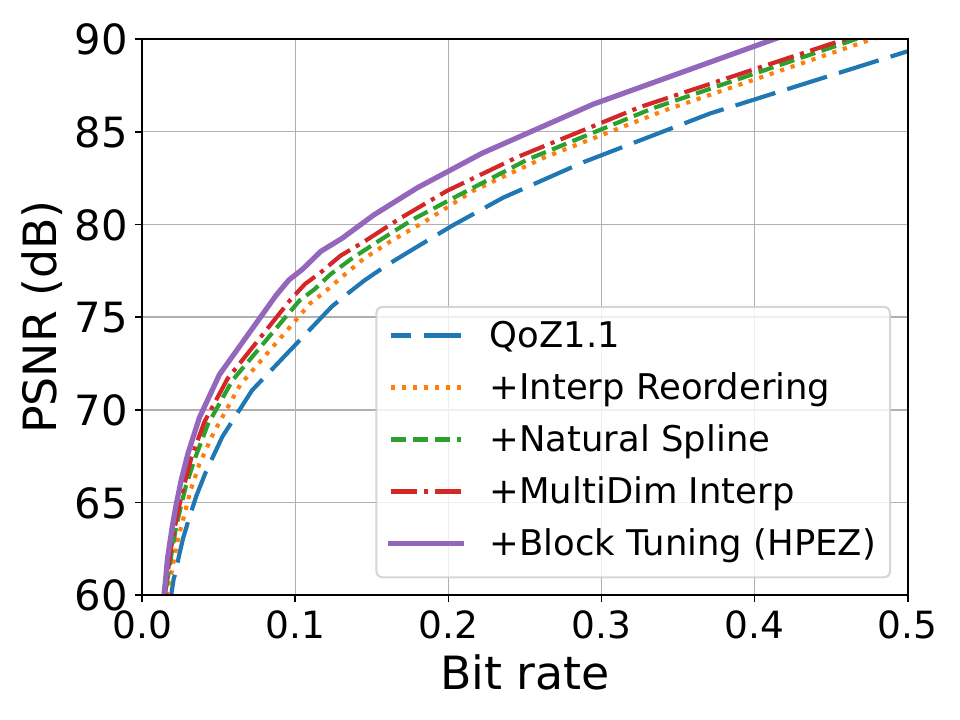}}
}
\hspace{-2mm}
\subfigure[{CESM-ATM}]
{
\raisebox{-1cm}{\includegraphics[scale=0.25]{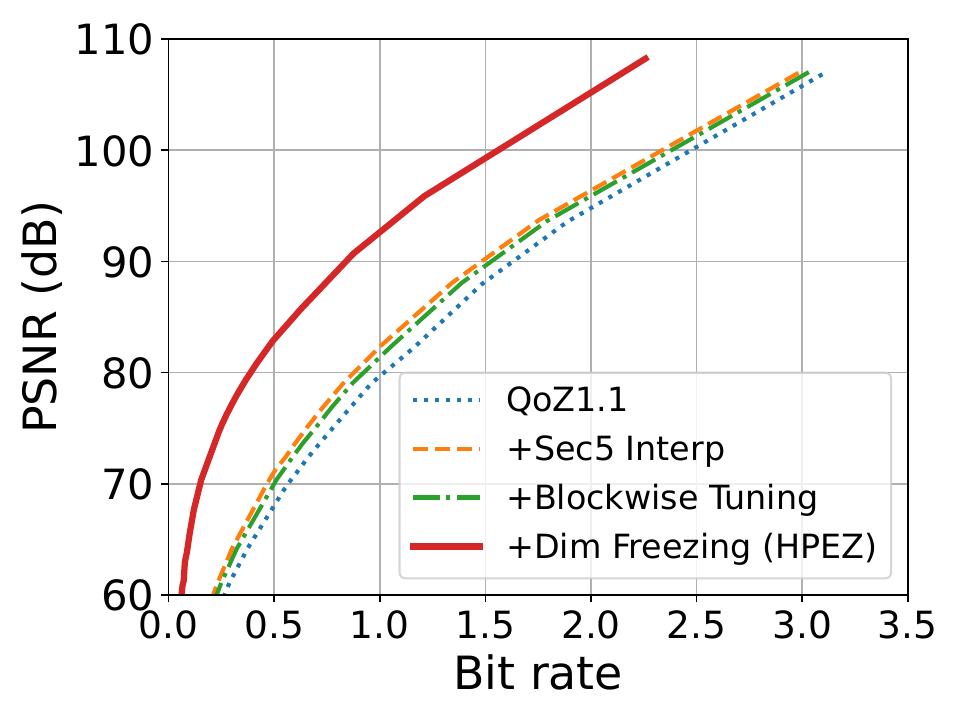}}
}
\hspace{-2mm}
\subfigure[{JHTDB}]
{
\raisebox{-1cm}{\includegraphics[scale=0.25]{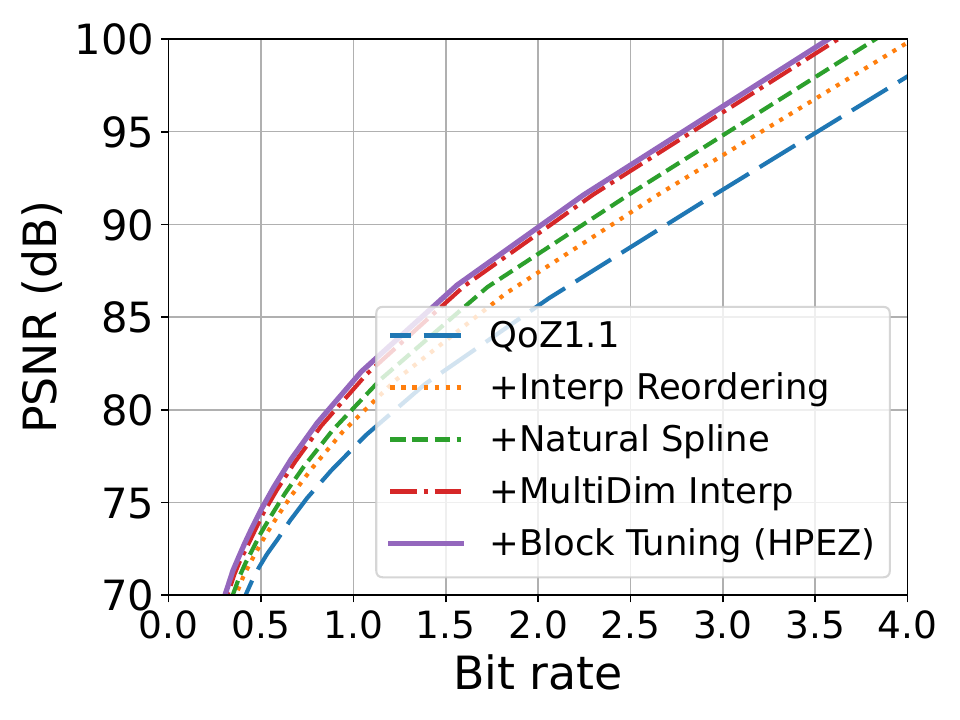}}
}
\hspace{-10mm}
\vspace{-4mm}

\hspace{-10mm}
\subfigure[{Miranda}]
{
\raisebox{-1cm}{\includegraphics[scale=0.25]{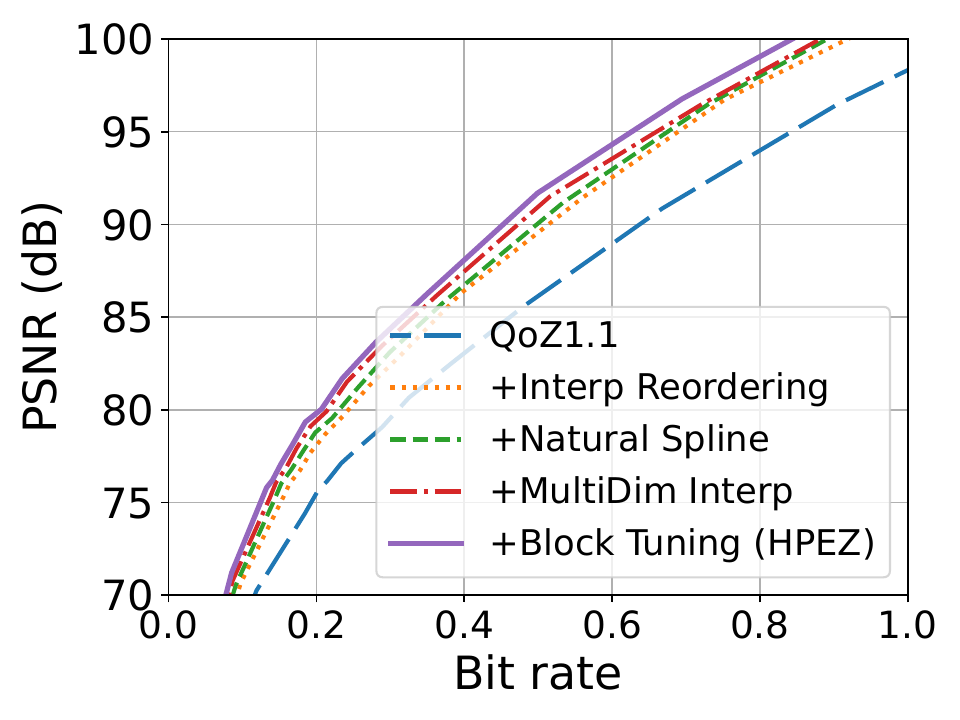}}%
}
\hspace{-2mm}
\subfigure[{SCALE-LetKF}]
{
\raisebox{-1cm}{\includegraphics[scale=0.25]{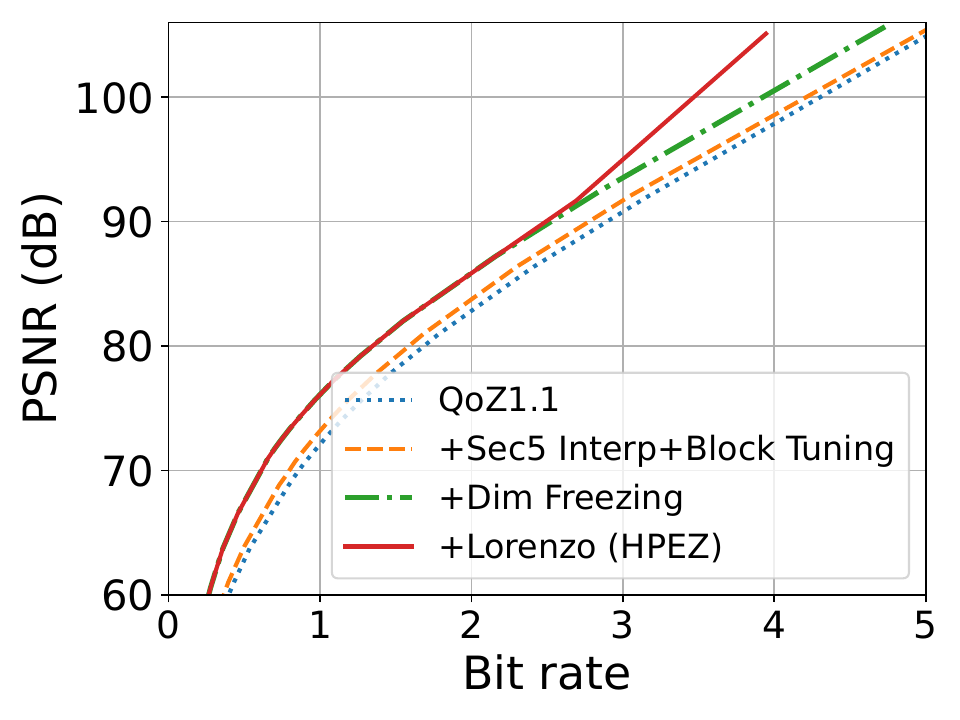}}%
}
\hspace{-2mm}
\subfigure[{SegSalt}]
{
\raisebox{-1cm}{\includegraphics[scale=0.25]{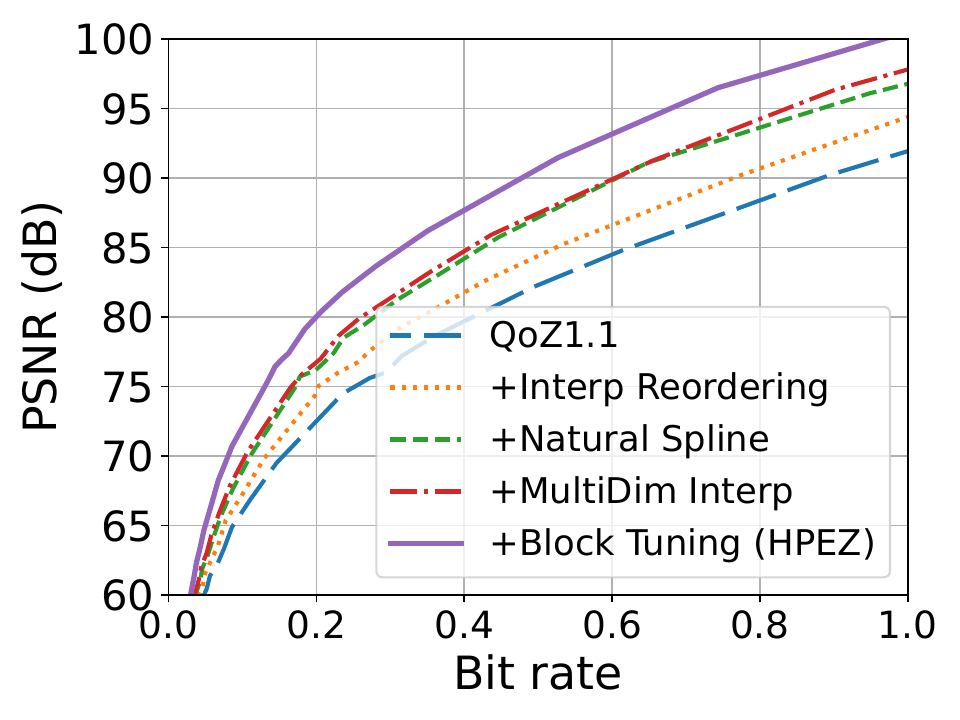}}%
}
\hspace{-10mm}
\vspace{-2mm}

\caption{Ablation study for rate-PSNR.}
\label{fig:eva-abla}
\end{figure}
Figure \ref{fig:eva-abla} shows the rate-PSNR plots 
on the RTM dataset with QoZ 1.1, HPEZ, and the different accumulations of new design components between them. For example, in Figure \ref{fig:eva-abla} (a) representing the results of the RTM dataset, there is a curve showing the rate-distortion of QoZ 1.1, a curve showing the rate-distortion of QoZ 1.1 plus the interpolation re-ordering, a curve for the aforementioned one plus the natural cubic spline, and so on, eventually to the complete HPEZ. For the RTM, JHTDB, Miranda, and SegSalt datasets (Figure \ref{fig:eva-abla} (a) (c) (d) (f)), analyzing the rate-distortion curves we can easily find that the HPEZ interpolation designs, including interpolation re-ordering (Section \ref{sec:reorder}), natural cubic spline (Section \ref{sec:natspline}), and multi-dimensional interpolation (Section \ref{sec:mdinterp}) all contribute to the improvement of rate-distortion. Additionally, block-wise interpolation tuning (Section \ref{sec:blockwise}) also plays an important role in optimizing the compression of their compression. Lastly, the effectiveness of multi-dimensional spline interpolation proved the generalization of Theorem \ref{theo:md} on diverse datasets and the integration of multiple interpolation optimization techniques.

In Figure \ref{fig:eva-abla} (b) and (e) corresponding to the CESM-ATM and SCALE-LetKF datasets, we can verify the effectiveness of dynamic dimension freezing (Section \ref{sec:freeze}) and the Lorenzo predictor (Section \ref{sec:lorenzo}) The dashed curve in Figure \ref{fig:eva-abla} (b) and (e) integrates all the interpolation designs in Section \ref{sec:interp}. Nevertheless, compared with QoZ 1.1 they have not refined the compression sufficiently. In contrast, the dynamic dimension freezing itself (the solid curves in Figure \ref{fig:eva-abla} (b) and the dash-dotted curve in Figure \ref{fig:eva-abla} (e)) has solely boosted the rate-distortion to a remarkable extent for those 2 datasets. Furthermore, comparing the solid curve and the dash-dotted curve in Figure \ref{fig:eva-abla} (e), leveraging the Lorenzo predictor has quite enhanced the compression quality of HPEZ in low-error-bound (i.e. high bit rates) cases.

Last, to examine the acceleration by the fast-varying-first interpolation described in Section \ref{sec:fvfi}, in Table \ref{tab:fvfi} we compare the sequential compression/decompression speeds of HPEZ between leveraging fast-varying-first interpolation or not (named as \textbf{HPEZ (w/o FVFI)}) in the table). Table \ref{tab:fvfi} clearly shows that the fast-varying-first interpolation has appreciably contributed to the performance of HPEZ, especially on the Miranda and JHTDB datasets.

\begin{table}[ht]
\centering
\footnotesize
  \caption {Compression Speeds (MB/s) with and without fast-varying-first interpolation ($\epsilon$=1e-3, i.e., $10^{-3}$)} 
  
  \label{tab:fvfi} 
  \begin{adjustbox}{width=0.75\columnwidth}
 
\begin{tabular}{|c|c|c|c|c|c|c|c|}
\hline
Type                           & Dataset         & CESM         & RTM          & Miranda      & SCALE        & JHTDB        & SegSalt      \\ \hline
\multirow{2}{*}{Cmp}   & HPEZ (w/o FVFI) & 132          & 139          & 101          & 124          & 87           & 134          \\ \cline{2-8} 
                               & HPEZ            & \textbf{140} & \textbf{142} & \textbf{140} & \textbf{129} & \textbf{105} & \textbf{141} \\ \hline
\multirow{2}{*}{Dcmp} & HPEZ (w/o FVFI) & 469          & 457          & 202          & 420          & 184          & 390          \\ \cline{2-8} 
                               & HPEZ            & \textbf{513} & \textbf{510} & \textbf{473} & \textbf{450} & \textbf{330} & \textbf{485} \\ \hline
\end{tabular}
\vspace{-2mm}
\end{adjustbox}
\end{table}
\section{Conclusion and Future Work}
\label{sec:conclusion}
In this paper, we propose HPEZ, an optimized interpolation-based error-bounded lossy compressor that supports quality-metric-driven auto-tuning and significantly improves compression ratio with low computation cost. 
The integration of advanced interpolation and auto-tuning designs in HPEZ has profoundly exploited the potential of the high-performance prediction-based compressor. In experiments, HPEZ achieves much better compression ratios and rate-distortion than existing high-performance error-bounded compressors with at most 140\% or 360\% compression ratio improvement under the same error bound or PSNR. HPEZ also over-performs existing error-bounded lossy compressors in data throughput tasks. In parallel data transmission experiments for distributed databases, HPEZ can achieve at most 40\% time cost reduction over the second bests, when compared with both high-performance and high-ratio error-bounded lossy compressors.

In the future, we plan to revise and develop HPEZ as follows: first, we will further optimize the speeds of HPEZ. Second, we will design more effective data prediction techniques for non-smooth data. Last, we will attempt to integrate compression techniques with a more flexible speed to adaptively tune the compression pipeline according to the requirements of compression speeds.

\section*{Acknowledgments}
This research was supported by the Exascale Computing Project (ECP), Project Number: 17-SC-20-SC, a collaborative effort of two DOE organizations – the Office of Science and the National Nuclear Security Administration, responsible for the planning and preparation of a capable exascale ecosystem, including software, applications, hardware, advanced system engineering and early testbed platforms, to support the nation’s exascale computing imperative. The material was supported by the U.S. Department of Energy, Office of Science, Advanced Scientific Computing Research (ASCR), under contract DE-AC02-06CH11357, and supported by the National Science Foundation under Grant OAC-2003709, OAC-2104023, OAC-2311875,  OAC-2311877, and OAC-2153451. We acknowledge the computing resources provided on Bebop (operated by Laboratory Computing Resource Center at Argonne).

\bibliographystyle{ACM-Reference-Format}

\bibliography{references}
\received{July 2023}
\received[revised]{October 2023}
\received[accepted]{November 2023}
\end{document}